\documentclass[a4 paper,11pt]{article} 
\usepackage[latin1]{inputenc}
\usepackage[english]{babel}

\usepackage{amsthm}
\usepackage{amsmath}
\usepackage{amssymb}
\usepackage{mathrsfs}
\usepackage{amsfonts}
\usepackage{bbm}
\usepackage{bigints}
\usepackage{dsfont}
\usepackage{mathtools}
\usepackage{stmaryrd}

\DeclareMathAlphabet{\mathpzc}{OT1}{pzc}{m}{it}

\usepackage{relsize}

\usepackage{setspace}

\usepackage{color}
\usepackage{graphicx}
\definecolor{orange}{rgb}{1,.5,0}
\definecolor{purple}{rgb}{0.65, 0, 1}
\definecolor{red}{rgb}{0.7,.1,0.2}

\usepackage[margin=1in]{geometry}

\usepackage{hyperref}
\hypersetup{colorlinks=false}

\usepackage[flushleft]{threeparttable}

\newtheorem{prop}{Proposition}

\newtheorem{cor}{Corollary}

\newtheorem{remark}{Remark}

\theoremstyle{definition}

\renewcommand{\qedsymbol}{\rule{0.7em}{0.7em}}

\usepackage{thmtools}

\makeatletter
\g@addto@macro\th@remark{\thm@headpunct{}}
\makeatother
\theoremstyle{remark}

\newtheorem{assump}{}

\declaretheoremstyle[headfont=\normalfont]{assump}

\usepackage{mathtools}

\usepackage{epsfig}
\usepackage{graphics}
\usepackage{epsfig}

\usepackage{multirow}
\usepackage{rotating}

\usepackage{pdfpages}

\usepackage{pgf}
\usepackage{tikz}
\usetikzlibrary{arrows,automata}
\usepackage[latin1]{inputenc}
\usepackage[position=top]{subfig}

\tikzstyle{vertex}=[circle,thick,draw=blue!75,fill=blue!20,scale=0.8]
\tikzstyle{selected vertex} = [vertex, fill=red!24]
\tikzstyle{edge} = [draw,thick,-]
\tikzstyle{weight} = [font=\small]
\tikzstyle{selected edge} = [draw,line width=5pt,-,red!50]
\tikzstyle{ignored edge} = [draw,line width=5pt,-,black!20]

\usepackage[misc]{ifsym}

\usepackage{authblk}

\usepackage{algorithm}
\usepackage{algorithmic}

\usepackage{epstopdf}
\usepackage{afterpage}

\usepackage[font=sf, labelfont={sf,bf}, margin=1cm, size=small]{caption}

\usepackage{multirow}
\usepackage{array}

\usepackage[round,comma, authoryear, sort&compress]{natbib}

\usepackage{enumitem}

\title{Discrete Responses in Bivariate Generalized Additive Models}
\author{Francesco Donat\thanks{Corresponding Author. University College London, Gower Street, London, WC1E 6BT, UK. Tel.: +44 (0)2076791223. E-mail: \href{mailto:f.w.donat@gmail.com}{f.w.donat@gmail.com}}\mbox{ }}
\author{Giampiero Marra}
\affil{\normalsize{Department of Statistical Science\\University College London}}

\date{}

\begin{document}

\maketitle

\begin{abstract}
\noindent A conceptual framework for the analysis of dichotomous and ordinal polychotomous responses within a penalized multivariate Generalized Linear Model is introduced. The proposed structure allows for a rather flexible predictor specification through the inclusion of non-parametric and spatial covariate effects, and the characterisation of the distribution of the stochastic model components with copulae of univariate marginals. Analytic derivations for the particular case of Gaussian marginals within a  bivariate system of dichotomous outcomes are also provided, and the framework is subsequently illustrated through the estimation of the HIV prevalence in Zambia using the 2007 DHS dataset.\vspace{0.2cm}

\end{abstract}

\section{Introduction}
Generalized Linear Models (GLMs, \citealp{NelderWedderburn_72}) are a comprehensive class of models that allows us to conduct estimation and inference for a variety of response types within the same coherent unifying framework. However, despite their undoubted relevance in applied research, they rely on a purely parametric specification of the covariate effects on the response, which effectively constraints the linear predictors to be a determined fixed-order polynomial, for instance. This is a strong requirement, as one cannot typically expect to know in advance the actual form of covariate-response relationships. This is especially the case in observational studies where their ``experimental'' situations are not conducted in a controlled manner. An actual risk for the researcher, therefore, would be to incorrectly specify the functional form of covariate effects, hence to potentially generate a non-negligible source of bias whenever these are not adequately represented. 

An existing approach to overcome this limitation is to consider a more flexible class of models that permits the representation and estimation of the additive effects of some continuous covariates of interest in a data-driven way. Methods of this kind are usually termed semi-parametric in the statistical literature (although this denomination is generally not shared by econometricians) because they conjugate both a parametric and a non-parametric characterisation of the functional forms of the regressors. Specifically, whenever the baseline structure is that of a GLM, the so-called Generalized Additive Models emerge (\citealp{HastieTibshirani_86,HastieTibshirani_90}), which complement their parametric counterparts by adopting a regression spline approach, implemented in a computationally stable and efficient manner by \cite{Wood_06}. Nonetheless, as any traditional regression analysis, GAMs are effectively models for the mean of a random variable possessing a certain conditional distribution function. To enhance flexibility, therefore, it is also licit to extend the framework to qualify the dependence of any moment of order higher than one on some explanatory variables of interest. In this way, the risks of mis-specifying the models and of conducting invalid inference from them is alleviated. This approach usually comes under the name of distributional regression, whose ideas have been variously incorporated within a GAM setting: for example, \cite{RigbyStasinopoulos_05} proposed a Generalized Additive Model for Location, Scale and Shape (GAMLSS), whose framework has been recently extended to the multivariate case by \cite{KKKL_15}. A review of these and of some other existing methodologies is presented in \cite{Kneib_13}. This line of research then seeks to achieve a higher degree of flexibility by increasing the number of distributions allowed by the proposed model representations, and by including in their respective specifications various kinds of covariate effects.

The present work aims at following these auspices in the context of discrete outcomes. Starting from the definition of a GAM for a $J$-variate vector of categorical responses as a penalized GLM, we discuss the conceptual representation of dichotomous and ordinal polychotomous dependent variables in terms of a triplet $(r,F_J,\mathbf Z)$, and of a penalty matrix $\mathbf S_{\boldsymbol\lambda}$ that allows us to incorporate in the model various instances, like non-parametric, spatial and random covariate effects. A method for dealing with a mixture of those two types of responses is also outlined. We then show how a generic estimation algorithm can be derived and inference subsequently conducted within the resulting multivariate Generalized Additive Model, and we argue that such algorithm can be, \textit{mutatis mutandis}, applied to any model representable in the $(r,F_J,\mathbf Z)$ form. Although the pace of the discussion is intentionally kept at a quite generic level, connections between the proposed framework and some existing models are made. These have the dual scope of motivating our representation with well-developed examples from the literature and, at the same time, of offering a way to extend them to the more flexible predictor specifications that form the domain of our work. In particular, attention is given to nested models accounting for unmeasured residual confounding: an instance rather frequent in observational studies and that may lead to detrimental consequences on the parameter estimates whenever it is not adequately controlled for (e.g.\ \citealp{Becher_92}). The proposed representation is then used to define a sample selection model for dichotomous responses to credibly assess the human immunodeficiency virus (HIV) prevalence in Zambia. With this empirical illustration, we give evidence of the flexibility of our generic representation which also permits the inclusion of multivariate distributions defined through copulae of univariate marginals, and the dependence of the corresponding association parameter to be expressed as a functional of the available data. In summary, therefore, this paper contributes to the literature by providing a flexible tool for the representation, estimation and inference in multivariate GLMs for discrete responses admitting non-parametric and spatial-dependent covariate effects, and by accounting for the unification of models for residual confounding under the same conceptual frame.

\section{A GAM Representation for Discrete Responses}\label{sct:GAMdiscr.GAMrep}
Let $\boldsymbol Y=(Y_1,\ldots,Y_J)^\top$ be a random vector with support the discrete set $\mathcal K:=\mathcal K_1\times\cdots\times\mathcal K_J$, where $\mathcal K_j:=\{1,\ldots,K_j\}$ and $\#(\mathcal K_j)=K_j<\infty$ for every $j\in\mathcal J$, $\mathcal J:=\{1,\ldots,J\}$; namely we consider each variable $Y_j$ to have finite $K_j$ levels. The set $\mathcal K_j$ is assumed here to collect both qualitative and quantitative elements, as well as variables measured on the nominal or ordinal scale (\citealp{Stevens_46}). Specifically, the former differentiates items based only on the categories they belong to, whereas the latter allows also for a rank order by which the realisations of $Y_j$ can be sorted, but still the relative degree of difference between them lacks of any meaningful interpretation. For notational convenience, we represent each $k_j$ by a natural number, $\mathcal K_j\subset\mathds N$, with the convention that, wherever the support of $Y_j$ is ordinal, we postulate the existence of an isomorphism that maps bijectively each element of the qualitative ordinal set $\mathcal K_j^*$ onto $\mathcal K_j$. In this case it holds: $\bar k_j^*\preceq k_j^*$ in $\mathcal K_j^*$ if and only if $\bar k_j\preceq k_j$ in $\mathcal K_j$, and we take the set $(\mathcal K_j,\preceq)$ to be totally ordered.

In analogy with the approach outlined in \cite{PTG_14} for the univariate case, we consider a regression of the probability $\pi_k=\mathbb P[\boldsymbol Y=k|\boldsymbol X=\boldsymbol x]$, with $k:=(k_1,\ldots,k_J)^\top\in\mathcal K$, on some covariates $\mathbf x:=\textup{vec}(\mathbf x_1,\ldots,\mathbf x_J)$ defined through a Generalized Linear Model form
\begin{equation}
\boldsymbol\pi=\boldsymbol g^{-1}(\boldsymbol\eta):=(\boldsymbol r^{-1}\circ\mathcal F)(\boldsymbol\eta_1,\ldots,\boldsymbol\eta_{K-1}),\label{eq:GAMdiscr.GLM}
\end{equation}
where $\boldsymbol r:\mathcal M\longrightarrow\mathcal P$ is a diffeomorphism from $\mathcal M:=\{(0,1)^{K-1}|\boldsymbol 1^\top\boldsymbol\pi<1\}$ to an open subset $\mathcal P$ of $(0,1)^{K-1}$, and with $\boldsymbol\pi:=\{\pi_k\}_{k\in\mathcal K\setminus\{K\}}$. Model (\ref{eq:GAMdiscr.GLM}) also comprises the map $\mathcal F:\mathcal S\longrightarrow\mathcal M$, with $\mathcal S\subset\mathds R^{K-1}$, and the array $\mathcal F(\boldsymbol\eta):=(F_J(\boldsymbol\eta_1),\ldots,F_J(\boldsymbol\eta_{K-1}))^\top$ is taken to collect fully-specified $J$-variate distribution functions, each of them evaluated at $\boldsymbol\eta_{k}:=(\eta_{1,k_1},\ldots,\eta_{J,k_J})^\top=\mathbf Z\boldsymbol\beta_k$, a linear predictor. Wherever needed, we assume that the elements of $\mathcal K$ obey a lexicographical order, that is $(\bar k_1,\ldots,\bar k_J)\preceq(k_1,\ldots,k_J)$ if and only if $\bar k_j\preceq k_j$ for all $j\in\mathcal J$ or $(\bar k_j=k_j\mbox{ }\wedge\mbox{ }\bar k_{\bar\jmath}\preceq k_{\bar\jmath}\mbox{ }\textup{for some}\mbox{ }\bar\jmath\in\mathcal J)$. A more traditional GLM representation for the $k$-th category can be recovered from (\ref{eq:GAMdiscr.GLM}) and reads as
\begin{equation}
r(\pi_k)=F_J(\boldsymbol\eta_k)=F_J(\mathbf Z\boldsymbol\beta_k),\label{eq:GAMdiscr.GLM2}
\end{equation}
where $r$ is now a function specific for the type of the responses $\boldsymbol Y$. For instance, in the univariate framework, dichotomous variables would set $r_j$ such that $\pi_{k_j}\mapsto\pi_{k_j}$, the identity map, therefore (\ref{eq:GAMdiscr.GLM2}) reduces to any model for the binary outcome, say a logit or a probit, depending on the definition of $F_1$. Models for ordinal polychotomous responses, as the Cumulative Link Model (CLM) of \cite{McCullagh_80}, also possess this representation and set the left-hand side of (\ref{eq:GAMdiscr.GLM2}) as $r_j(\pi_{k_j})=\pi_1+\cdots+\pi_{k_j}$. Although more specifications of different univariate response types are illustrated in \cite{PTG_14}, in this work we confine ourselves to the sole study of dichotomous and ordinal outcomes, since they are the most frequent instances of the class of models we aim at developing. 

As (\ref{eq:GAMdiscr.GLM2}) explicates, any GLM for discrete responses is fully characterised by the triplet $(r,F_J,\mathbf Z)$, where the design matrix $\mathbf Z$ depends on the covariates $\mathbf x$, though not necessarily coinciding with them. For example, let the polychotomous response $\boldsymbol Y$ follow the model
\begin{equation}
r(\pi_k)=\sum_{\tilde k_1\le k_1}\cdots\sum_{\tilde k_J\le k_J}\pi_{\tilde k_1,\ldots,\tilde k_J}=F_J(\boldsymbol c_k-\mathbf X\boldsymbol\beta)=F_J(\mathbf Z\boldsymbol\beta_k),\label{eq:GAMdiscr.GLMOrdinal}
\end{equation}
then $\mathbf Z:=\textup{diag}(\mathbf z_1^\top,\ldots,\mathbf z_J^\top)$ and $\boldsymbol\beta_k:=\textup{vec}(\boldsymbol\beta_{1,k},\ldots,\boldsymbol\beta_{J,k})$, where $\mathbf z_j:=(1,-\textup x_{j,1},\ldots,-\textup x_{j,m_j})^\top$ and $\boldsymbol\beta_{j,k}:=(c_{j,k_j},\beta_{j,1},\ldots,\beta_{j,m_j})^\top$. In the proceeding analysis, we call the $c_{j,k_j}$'s threshold parameters or cut points, and we assume they are the only elements in the corresponding linear predictor $\eta_{j,k_j}$ to depend on the categories of $Y_j$. We also stress that only $K_j-1$ cut points are effectively estimable in this framework because, in order to allow the domain of $F_J$ to coincide with the extended real hyper-plane $\mathds R^J$, we need to impose $c_{j,K_j}=\infty$ for any $j$ and $c_{j,0}:=c_{j,1-1}=-\infty$. As a consequence, a dichotomous response with support $\mathcal K_j:=\{0,1\}$ would set the only threshold to $0$, and the model intercept is now estimable.

\paragraph{Ordinal Polychotomous Outcomes}
This instance is of some interest in terms of the proposed GLM specification and worth to be discussed further. Notice first that the given definition of $r(\pi_k)$ is posing a constraint to the set $\mathcal P$. Specifically, by assuming $\bar k_j\preceq k_j$, we have
\begin{eqnarray}
r(\pi_k)&=&\sum_{\tilde k_1\le k_1}\cdots\sum_{\tilde k_J\le k_J}\pi_{\tilde k_1,\ldots,\tilde k_J}\nonumber\\
&=&\sum_{\tilde k_1\le k_1}\cdots\sum_{\tilde k_j\le\bar k_j}\cdots\sum_{\tilde k_J\le k_J}\pi_{\tilde k_1,\ldots,\tilde k_J}+\sum_{\tilde k_1\le k_1}\cdots\sum_{\tilde k_j\in[\bar k_j,k_j]}\cdots\sum_{\tilde k_J\le k_J}\pi_{\tilde k_1,\ldots,\tilde k_J}\nonumber\\
&=&r(\pi_{\bar k})+r'(\pi_k)\ge r(\pi_{\bar k})\nonumber
\end{eqnarray}
since $r'(\pi_k)$ is the sum of probability measures. We also deduce $\bar k:=(k_1,\ldots,\bar k_j,\ldots,k_{J-1})\preceq (k_1,\ldots,k_j,\ldots,k_J)=:k$ by the assumed lexicographical order. Hence $r(\pi_{\bar k})\le r(\pi_k)$ for any $\bar k\preceq k$, and $\mathcal P:=\{\boldsymbol r\in(0,1)^{K-1}|r(\pi_{\bar k})\le r(\pi_k),\textup{ for all }\bar k\preceq k\textup{ and }\bar k,k\in\mathcal K\}$. 

If this restriction comes from the very construction of a CLM, a second one emerges to let the model meet a general coherency condition. To establish this result, the inspection of (\ref{eq:GAMdiscr.GLMOrdinal}) reveals that the linear predictors $\{\boldsymbol\eta_k\}$ depend on the element $k$ of the discrete ordered set $\mathcal K$ that one attempts to model. The sought coherency requires, therefore, the definition of a specific correspondence between the order relations existing in $k\in\mathcal K$ with those in $\boldsymbol\eta_k\in\mathcal S$. This is identified, in particular, in the order embedding of each $\mathcal K_j$ into a relevant subset of the real line as induced by the thresholds $\{c_{j,k_j}\}$: in this way, it is possible to construct non-overlapping hyper-rectangles in $\mathds R^J$ isomorphic to $(k_1,\ldots,k_J)\in\mathcal K$. In terms of a multivariate CLM, the order embedding is guaranteed by taking the cut points $\{c_{j,k_j}\}$ to be an increasing sequence in $k_j$ for every $j\in\mathcal J$ wherever, as stated above, the threshold parameters are the only quantities in the linear predictors depending on the categories of $Y_j$. Consequently, it follows the isomorphism $\{\boldsymbol\eta_k\}\cong\{k\}$, meaning that there exists a bijection $\varphi:\mathcal K\longrightarrow\mathcal S$ such that $\varphi(\bar k):=\boldsymbol\eta_{\bar k}\preceq\boldsymbol\eta_k=:\varphi(k)$ in $\mathcal S$ if and only if $\bar k\preceq k$ in $\mathcal K$. In this case, the domain of $\mathcal F$ is then restricted to be the set $\mathcal S:=\{\boldsymbol\eta\in\mathds R^{K-1}|\boldsymbol\eta_{\bar k}\preceq\boldsymbol\eta_k,\textup{ for all }\bar k\preceq k\textup{ and }\bar k,k\in\mathcal K\}$. In the bivariate case, $J=2$, \cite{Dale_86} imposed a strict monotonicity on the cut points to imply a non-degenerate probability measure on $\mathcal K$. Although this condition would in turn debar one possible source of the Maximum Likelihood Estimator to be located at the boundary of the parameter space with non-null probability (refer to \citealp{Haberman_80} for the univariate case), we reckon this restriction has to be avoided as it arbitrarily excludes a still admissible estimate, albeit at the boundary. Arguably, two congruent subsequent cut points are not in contrast with the coherency principle. In fact, given any two elements $\bar k$ and $k$ of $\mathcal K$ such that either $\bar k_j=k_j$ or $\bar k_j\prec k_j$, the coherency implies that the occurrence $c_{j,\bar k_j}=c_{j,k_j}$ is verified if and only if $\bar k_j=k_j$, that is whenever $\mathcal K:=\mathcal K_1\times\cdots\times(\mathcal K_j\setminus\{\bar k_j\})\times\cdots\times\mathcal K_J$, unless $\bar k_j$ is an element of zero probability mass in $\mathcal K$. Either cases correspond to observing zero counts for the $k_j$-th category, but with the coherency would be still in place.

Under this principle, it is possible to motivate ordinal polychotomous responses through a generating continuous latent random vector $\boldsymbol Y^*:=(Y_1^*,\ldots,Y_J^*)^\top$ in $\mathds R^J$ in such a way that, upon letting $\boldsymbol\epsilon$ be the stochastic component in the regression $\boldsymbol Y^*=\mathbf X\boldsymbol\beta+\boldsymbol\epsilon$, it holds
\begin{equation}
\{\boldsymbol Y\preceq k\}\Longleftrightarrow\{\boldsymbol\epsilon\le\boldsymbol\eta_k\},\label{eq:GAMdiscr.manifestLatent}
\end{equation}
where the right-hand side is intended component-wise as $\epsilon_j\le\eta_{j,k_j}$ for every $j$.

\subsection{Specification of the Linear Predictors in a Penalized GLM}\label{sct:GAMdiscr.GLMpenal}
Linear predictors enter representations (\ref{eq:GAMdiscr.GLM}) or (\ref{eq:GAMdiscr.GLM2}) as domain of the distribution functions collected in $\mathcal F$, and are fully characterised by the design matrix $\mathbf Z$ for any vector of parameters $\boldsymbol\beta_k$. In the proceeding, we assume that each predictor $\eta_{j,k_j}$ depends parametrically on some covariates $\mathbf x_j$, and through an additive form of unknown smooth functions $s_{j,l_j}:\mathds R\longrightarrow\mathds R$ for the remaining continuous regressors $\mathbf v_j:=(\textup v_{j,1},\ldots,\textup v_{j,L_j})^\top$. The resulting functional form is then termed semi-parametric, and defines the class of (Vector) Generalized Additive Models (VGAMs; \citealp{YeeWild_96}). We prefer, however, to adopt the alternative terminology of penalized GLMs as, in our opinion, it reflects better the features of the class of models we discuss beyond the traditional domain of GAMs. 

For dichotomous and ordinal polychotomous responses we define the linear predictor to be
\begin{equation}
\eta_{j,k_j}(\mathbf x_j,\mathbf v_j,\vartheta)=c_{j,k_j}-\mathbf x_j^\top\boldsymbol\beta_j-s_{j,1}(\textup v_{j,1})-\cdots-s_{j,L_j}(\textup v_{j,L_j})\mbox{ }\mbox{ }\mbox{ }\mbox{ }\mbox{ }\textup v_{j,l_j}\in\mathds R,\label{eq:GAMdiscr.linPred}
\end{equation}
where $\boldsymbol\vartheta$ denotes the parameter vector, and the smooth functions are represented by regression splines using the approach popularised in the literature by \cite{EilersMarx_96}. Assume first that we have a sample of $n$ observations indexed by the subscript $i$. The underlying idea of the method is to approximate each curve by a linear combination of known basis spline functions, $b_{j,l_j,h_j}$, for $h_j=1,\ldots,H_j$, and unknown regression parameters to be estimated within the system, $\delta_{j,l_j,h_j}$. In our notation, $h_j$ is employed to count the bases, as delimited by some knot points in the interior of $[\textup v_{j,l_j,(1)},\textup v_{j,l_j,(n)}]$ for every $j$. Upon defining $\boldsymbol b_{j,l_j}(\textup v_{j,l_j,i}):=(b_{j,l_j,1}(\textup v_{j,l_j,i}),\ldots,b_{j,l_j,H_j}(\textup v_{j,l_j,i}))^\top$, and $\boldsymbol\delta_{j,l_j}$ the corresponding $H_j$-dimensional vector of parameters associated with the smooths, it holds 
\begin{equation}
s_{j,l_j}(\textup v_{j,l_j,i})\approx\boldsymbol\delta_{j,l_j}^\top\boldsymbol b_{j,l_j}(\textup v_{j,l_j,i}).\nonumber
\end{equation}
The evaluation of $\boldsymbol b_{j,l_j}$ for each $i$ yields $H_j$ curves -- encompassing different degrees of complexity -- that give, once multiplied by some real-valued parameter vector and then summed, an estimated curve for $s_{j,l_j}$. Basis functions are usually chosen to have convenient mathematical properties and good numerical stability: possible instances are B-splines, cubic regression and low-rank thin plate regression splines (e.g.\ \citealp{RWC_03} and \citealp{Wood_03}). For identifiability purposes, a centering constraint such as $\sum_i s_{j,l_j}(\textup z_{j,l_j,i})=0$ for every $l_j$ has to be imposed, which is automatically incorporated in our model representation using the parsimonious approach of \cite{Wood_06}.

We are now in the position to express the functional form of linear predictors (\ref{eq:GAMdiscr.linPred}) through a more compact and comprehensive representation. To this end, let us define $\boldsymbol\beta_{[j,l_j]}:=\boldsymbol\delta_{[j,l_j]}$ the sub-vector of $\boldsymbol\beta$ corresponding to the $(j,l_j)$-th smooth and, accordingly, $\mathbf X_{[j,l_j]}$ the covariate matrix whose $i$-th row is $\boldsymbol b_{j,l_j}^\top(\textup v_{j,l_j,i})$. It then follows that the $n$-dimensional vector of linear predictors for the $j$-th response can be written as:
\begin{equation}
\boldsymbol\eta_j:=\boldsymbol c_j-\mathbf X_{j,1}\boldsymbol\beta_{j,1}-\cdots-\mathbf X_{j,M_j}\boldsymbol\beta_{j,M_j}=\mathbf Z_j\boldsymbol\beta_j,\label{eq:GAMdiscr.predOrdinal}
\end{equation}
where $\mathbf Z_j:=(\boldsymbol 1_n,-\mathbf X_{j,1},\ldots,\mathbf X_{j,m_j},\ldots,\mathbf X_{j,M_j})$ and $\boldsymbol\beta_j:=\textup{vec}(\boldsymbol c_{j,k_j},\boldsymbol\beta_{j,1},\ldots,\boldsymbol\beta_{j,M_j})$. We further assume $j=1,\ldots,\bar J$, with $\bar J\ge J$ to possibly specify any association parameter implied by the distribution $F_J$ in terms of some observed independent variables. So re-written, the linear predictors conform notationally with the GLM given in (\ref{eq:GAMdiscr.GLMOrdinal}), with the caveat that now they can be indifferently used to represent linear and non-linear covariate effects within a generic GLM for multivariate discrete data. To see this, set $\boldsymbol\eta:=(\boldsymbol\eta_1,\ldots,\boldsymbol\eta_{\bar J})$ to be the array whose $i$-th row is $\boldsymbol\eta_k=(\eta_{1,k_1},\ldots,\eta_{J,k_J},\ldots,\eta_{\bar J})^\top$ and, accordingly, $\boldsymbol\beta_k:=\textup{vec}(\boldsymbol\beta_1,\ldots,\boldsymbol\beta_{\bar J})$. Thus $\mathbf Z_i:=\textup{diag}(\mathbf z_{1,i}^\top,\ldots,\mathbf z_{\bar J,i}^\top)$, with $\mathbf z_{j,i}^\top$ being the $i$-th row of $\mathbf Z_j$, and the linear form that defines the right-hand side of (\ref{eq:GAMdiscr.GLMOrdinal}) is now recovered as $\boldsymbol\eta_k:=\mathbf Z_i\boldsymbol\beta_k$.

\paragraph{Characterisation and Definition of a GAM as a Penalized GLM}
To enforce certain properties of the covariate effects, a ridge-type penalisation is assigned to each column of $\mathbf Z_j$, namely $\mathcal P_{j,m_j}:=\lambda_{j,m_j}\boldsymbol\beta_{j,m_j}^\top\overline{\mathbf S}_{j,m_j}\boldsymbol\beta_{j,m_j}$, where the dimension of $\overline{\mathbf S}_{j,m_j}$ is generically denoted by $q$. The smoothing (or tuning) parameter $\lambda_{j,m_j}\in[0,\infty)$, in particular, is introduced here to control the trade-off between smoothness and fitting in the non-parametric estimation of $s_{j,m_j}$. Specifically, as $\lambda_{j,m_j}$ tends to zero, less penalisation is attached to the regression coefficient $\boldsymbol\beta_{j,m_j}$ and the estimation occurs either at the pre-specified polynomial form for the parametric model components, or at the spline interpolation for the unknown functions. On the other hand, an infinite value of $\lambda_{j,m_j}$ results in the fitting of a straight line, a situation also known as over-smoothing.

The penalty $\mathcal P_{j,m_j}$ can be used to describe several covariate effects in the same unifying manner. In particular, a parametric functional form would set $\overline{\mathbf S}_{j,m_j}=\boldsymbol 0_{q,q}$, and the corresponding $\mathbf X_{j,m_j}\boldsymbol\beta_{j,m_j}$ reduces to $\mathbf x_{j,m_j}\beta_{j,m_j}$, with $\mathbf x_{j,m_j}:=(\textup x_{j,m_j,1},\ldots,\textup x_{j,m_j,n})^\top$ whereas, in the presence of non-parametric effects, one can specify the penalty through the symmetric and positive semi-definite matrix
\begin{equation}
\overline{\mathbf S}_{j,m_j}:=\int_{V_{j,m_j}}\boldsymbol b_{j,m_j}''(\boldsymbol b_{j,m_j}'')^\top\textup d\textup v_{j,m_j},\nonumber
\end{equation}
a generic measure of the curvature of the estimated $(j,m_j)$-th smooth function (see \citealp{GreenSilverman_94} for a detailed introduction to this roughness penalty approach to curve estimation).

Furthermore, $\mathcal P_{j,m_j}$ is compatible with the specification of random effects models ($\overline{\mathbf S}_{j,m_j}=\mathbf I_q$), as well as with the definition of spatial covariate dependence. As recently outlined by \cite{KKKL_15} and \cite{MRBWM_15}, this approach can be employed wherever the phenomenon of interest to possesses characteristics that vary according to the geographical location of each individual. Let $r_i\in\mathcal R$, $\mathcal R:=\{1,\ldots,R\}$, be the region to which the $i$-th observation belongs, and define $\mathbf R$ the corresponding $n\times R$ design matrix of the spatial effects. This sets $\mathbf X_{j,m_j,i}\boldsymbol\beta_{j,m_j,i}\equiv\mathbf x_{r_i}^\top\boldsymbol\beta_{r_i}$ so that we estimate separate parameters $\boldsymbol\beta_1,\ldots,\boldsymbol\beta_R$ for each region, and $\mathbf R$ is an incidence matrix, namely an array such that $\mathbf R_{[i,r]}=1$ if observation $i$ belongs to $r\in\mathcal R$, and 0 otherwise. Then, for discrete spatial effects, a Markov random field (e.g.\ \citealp{RueHeld_05}) induces the penalty
\begin{equation}
\overline{\mathbf S}_{j,m_j}\equiv\overline{\mathbf S}_{[r,s]}=\left\{\begin{array}{cl}-1 & r\neq s\mbox{ }\wedge\mbox{ }s\in\delta_r\\
0 & r\neq s\mbox{ }\wedge\mbox{ }s\not\in\delta_r\\
N_r & r=s\end{array}\right.,\nonumber
\end{equation}
with $\delta_r$ denoting the set of neighbours of region $r$, and $N_r$ the number of regions in $\delta_r$.

Once every component in $\mathbf Z_j$ is endowed with a proper penalisation depending on the desired effect one is willing to model, and $\overline{\mathbf S}_{j,m_j}$ adequately adjusted to meet the splines' centering constraint, an overall penalty for the whole model can be constructed as $\mathcal P_{\boldsymbol\lambda}:=\boldsymbol\vartheta^\top\mathbf S_{\boldsymbol\lambda}\boldsymbol\vartheta$, where $\mathbf S_{\boldsymbol\lambda}$ corresponds to $\overline{\mathbf S}_{\boldsymbol\lambda}$ padded with zeros so that $\boldsymbol\vartheta^\top\mathbf S_{\boldsymbol\lambda}\boldsymbol\vartheta=\boldsymbol\beta^\top\overline{\mathbf S}_{\boldsymbol\lambda}\boldsymbol\beta$, $\overline{\mathbf S}_{\boldsymbol\lambda}:=\textup{diag}\{\overline{\mathbf S}_{j,m_j}\}_{m_j,j}$ and $\boldsymbol\lambda:=\{\lambda_{j,m_j}\}_{m_j,j}$. A penalized GLM is therefore defined as any model in the form of (\ref{eq:GAMdiscr.GLM2}) augmented with a non-zero penalty $\mathcal P_{\boldsymbol\lambda}$.

In the next section, we qualify the generic framework to describe a class of models widely used in applications, and we show how it can be represented within the $(r,F_J,\mathbf Z)$ frame. In this way, these models can be extended beyond the parametric specification of their functional form of covariate effects, and their estimation and inference will then be a direct consequence of those of a generic multivariate penalized GLM for discrete responses.

\subsection{Some Bivariate Models in the Class of Penalized GLMs}\label{sct:GAMdiscr.ConfoundingSS}
The analysis of observational data may be difficult as they often depart from the ideal conditions underlying any (also rather simple) regression model. They are commonly characterised by a lack of randomisation that may result either in a non-random selection of individuals in the sample, or even in the non-random allocation of a predictor of interest among the population (hence inducing a distorted association with the outcome). The former is commonly referred to non-random sample selection, and arises whenever individuals select themselves in or out of the relevant sample. It is often the case that some factors that determine the membership to the selected sample are also associated with those that determine the outcome itself. In the empirical illustration accompanying this work, and concerning the estimation of the HIV prevalence in Zambia, the refusal of people to be tested for the virus may be induced by factors associated to their HIV status. For example, they may already know or correctly predict their seropositivity and so fear others will learn about it if tested. The latter instance is regarded instead as a form of endogeneity as it is denominated in the econometric literature, and it may stem from different sources, including, but not limited to, the direct unmeasured confounding problem; \cite{Wooldridge_02} discusses this in detail as well as other generating sources of endogeneity, and we refer to him for a more thoughtful illustration of the topic. This situation arises whenever a common background variable affects simultaneously both the outcome of interest and one of its regressors, and it is not readily observable or quantifiable by the researcher. The affected covariate is then termed endogenous, and its effect on the outcome results confounded. A pedagogical example is the estimation of the effect of education on wages. Both the relevant variables in this study can be co-determined by factors such as personal ability and motivation that are likely to be explainable to both individual's level of education and salary, but hardly measurable (see for example \citealp{Imbens_14} for an interesting survey on the topic).

When not accounted for, non-random sample selection and endogeneity can both lead to inconsistent parameter estimates for the whole model. To deal with these issues, in some early works, \cite{Heckman_78, Heckman_79} devised a two-step estimation procedure for a prototypical recursive bivariate system of equations in a dichotomous responses setting, with $\boldsymbol Y=(Y_1,Y_2)^\top$ and $Y_2(Y_1)$. His proposals specified a binary rule for the observability of the outcome of interest, $Y_2$, for the non-random sample selection case, and related the conditional mean of the endogenous regressor, $Y_1$, to various other predictors when endogeneity is suspected. In either scenario, identification of the true association between the elements of $\boldsymbol Y$ would require to be able to qualify the dependence of $Y_1$ on a relevant variable which is assumed to be independent of both $Y_2|Y_1$ and the unmeasured confounder(s).

\paragraph{Unmeasured Confounding}
We consider the case where both the responses are discrete, and we specify the following triangular generating structure for the $j$-th categorical response in terms of a latent variable formulation, as of Case 3 in \cite{Heckman_78},
\begin{equation}
Y_j^*=\mathds 1_{j=2}\{\psi Y_1^*\}+\mathbf x_{j,1}^\top\boldsymbol\beta_{j,1}+\cdots+\mathbf x_{j,M_j}^\top\boldsymbol\beta_{j,M_j}+\epsilon_j\mbox{ }\mbox{ }\mbox{ }\mbox{ }\mbox{ }(\epsilon_1,\epsilon_2)^\top\sim\mathcal N_2(\boldsymbol0,\boldsymbol\Omega;\rho),\label{eq:GAMdiscr.triangular}
\end{equation}
where $\mathbf x_{j,m_j}^\top$ denotes the $(j,m_j)$-th row of $\mathbf X_{j,m_j}$. The given distributional assumption is in line with the considerations of \cite{GreeneHensher_10} and with the current practice in multivariate discrete response modelling. In particular, we set $\boldsymbol\Omega:=\mathbb V\textup{ar}[\boldsymbol\epsilon]$ to have unit main diagonal elements for identifiability purposes and correlation coefficient $\rho$. Let us next define the following quantities:
\begin{equation}
\boldsymbol\Gamma:=\left[\begin{array}{cc}1 & 0\\ -\psi & 1\end{array}\right]\mbox{ }\mbox{ }\mbox{ }\mbox{ and }\mbox{ }\mbox{ }\mbox{ }\mathbf L:=\left[\begin{array}{cc}1 & 0\\ 0 & \sqrt{1+2\psi\rho+\psi^2}\end{array}\right],\nonumber
\end{equation}
where $\mathbf L$ is a lower triangular, positive definite matrix since $(1+2\psi\rho+\psi^2)=(1-\rho^2)+(\rho+\psi)^2>0$, and with $\mathbf L\boldsymbol Y^*=\mathbf L\boldsymbol\Gamma^{-1}(\mathbf X\boldsymbol\beta+\boldsymbol\epsilon)$ distributed as Standard Normal random vector with covariance matrix $\mathbf\Sigma=\mathbf L\boldsymbol\Gamma^{-1}\boldsymbol\Omega\boldsymbol\Gamma^{-1}\mathbf L^\top$. This structure constitutes the most general one we discuss in this paper, as it nests the vast majority of the other model specifications for unmeasured confounding and non-random sample selection currently proposed in the literature. Notice that the manifest polychotomous ordinal responses can now be recovered from (\ref{eq:GAMdiscr.manifestLatent}) via the series of equivalences
\begin{equation}
\{\boldsymbol Y\preceq k\}\Longleftrightarrow\{\boldsymbol\Gamma\boldsymbol Y^*\le\boldsymbol c_k\}\Longleftrightarrow\{\mathbf L\boldsymbol\Gamma^{-1}\boldsymbol\epsilon\le\mathbf L\boldsymbol\Gamma^{-1}\mathbf Z\boldsymbol\beta\},\nonumber
\end{equation}
which implies that the predictor $\mathbf L\boldsymbol\Gamma^{-1}\boldsymbol\eta_k$ is non-linear in $\boldsymbol\vartheta$, an occurrence that has to be accounted for in the derivation of the estimation algorithm. The corresponding $(r,F_2,\mathbf Z)$ definition of the triangular form (\ref{eq:GAMdiscr.triangular}) for every $k\in\mathcal K$ is then
\begin{equation}
\left(\sum_{\bar k_1\le k_1}\sum_{\bar k_2\le k_2}\pi_{\bar k_1,\bar k_2},\Phi_2(\eta_{1,k_1},\eta_{2,k_2};\boldsymbol\Sigma),\mathbf L\boldsymbol\Gamma^{-1}\mathbf Z_i\right),\label{eq:GAMdiscr.rFZtriang}
\end{equation}
where $\mathbf Z$ is defined as in (\ref{eq:GAMdiscr.predOrdinal}) and, in the case of $\mathcal K:=\{0,1\}\times\{0,1\}$, the first entry of the triplet becomes $r(\pi_k)=\pi_{k_1,k_2}$ while the design matrix has to be modified accordingly in order to accommodate the instance $c_{j,0}=0$ for any $j$. This representation of the recursive structure allows explicitly for a latent endogenous predictor to be a determinant of the intentions about $Y_2$. That is, if one interprets the intentions towards a manifest discrete outcome (the actual action) as the result of an underlying choice mechanism as described by $Y_j^*$, then (\ref{eq:GAMdiscr.rFZtriang}) really assumes the existence of unobservables that influence simultaneously the intentions about the components of $\boldsymbol Y$. For example, there is a vast economic literature pointing out that the choice of investing in both health and education are confounded by individual time preferences, in the sense that people with low (high) rates of time preference are more (less) likely to decide to invest in both schooling and health (\citealp{Sander_95}, \citealp{Fuchs_82}, \citealp{VanderPol_11}). In this case, endogeneity is regarded to act at the level of the choice of how much to invest in education and future health status. A researcher could nonetheless be interested in modelling the effects of an observed endogenous variable (where the intentions have been revealed by the actual choices undertaken) on the discrete response $Y_2$: this instance then specifies $\boldsymbol\Gamma=\mathbf L=\mathbf I_2$ and $\mathbf Z$ is assumed to include the levels of the manifest $Y_1$. A discussion about the distinctive features of these different modelling strategies can be found in \cite{Vossmeyer_14}, who also introduced a formal Bayesian model comparison framework to test these two competing models against the observed data.

Models for a mixture of dichotomous and ordinal polychotomous responses can also be reconciled within this representation by giving a proper definition of $r(\pi_k)$. In fact, this is the only element which is directly affected by the types of responses considered, whereas the design matrix mainly attains the functional form assumed for the predictors, and $F_2$ specifies the link function. To describe this situation, we may think of $\boldsymbol r$ as the composition function $\boldsymbol r=\boldsymbol r_1\circ\boldsymbol r_2$, where the subscripts correspond to the elements of the $2$-dimensional vector $\boldsymbol Y$ they refer to. Therefore, by setting
\begin{equation}
r_j(\pi_k)=\pi_{k_j,k_{\bar\jmath}}\mbox{ }\mbox{ }\mbox{ }\mbox{ and }\mbox{ }\mbox{ }\mbox{ }r_{\bar\jmath}(\pi_k)=\sum_{\tilde k_{\bar\jmath}\le k_{\bar\jmath}}\pi_{k_j,\tilde k_{\bar\jmath}},\nonumber
\end{equation}
we have $r_{\bar\jmath}(r_j(\pi_k))=r_{\bar\jmath}(\pi_k)$ for any $\bar k,k\in\mathcal K$ and $\bar\jmath:=\mathcal J\setminus\{j\}$. The result follows because, for dichotomous responses, $r_j$ is the identity map, meaning that it is also indifferent the order in which the types of the outcomes appear in $\boldsymbol Y$ in terms of the model representation. In other words, irrespectively from whether $Y_1$ or $Y_2$ is the dichotomous variable, the function composition $(r_j\circ r_{\bar\jmath})$ is commutative:
\begin{equation}
(r_{\bar\jmath}\circ r_j)(\pi_k)=r_{\bar\jmath}(\pi_k)=r_j(r_{\bar\jmath}(\pi_k))=(r_j\circ r_{\bar\jmath})(\pi_k).\nonumber
\end{equation}

\paragraph{Non-random Sample Selection}
In this case, it is assumed that the outcome $Y_2\in\mathcal K_2$ is observed if and only if $\{0,1\}\ni Y_1=1$, whereas it is labelled as missing otherwise. As a consequence, the vector $\boldsymbol\pi$ results further constrained: every element in the form of $\pi_{0,k_2}$ is now not a sensible quantity in the model for any $k_2$, since it refers to a missing value in the realisation of $Y_2$. Hence, one can only describe the corresponding marginal probability $\pi_{0\cdot}$, which is translated mathematically into the map $\mathcal M\longrightarrow\mathcal M^s$ defined by $\pi_{0,k_2}\mapsto\sum_{\tilde k_2\in\mathcal K_2}\pi_{0,\tilde k_2}=:\pi_{0\cdot}$, and $\pi_{1,k_2}\mapsto\pi_{1,k_2}$ for any $k_2\in\mathcal K_2\setminus\{K_2\}$, where $\mathcal M^s:=\{\boldsymbol\pi^s\in(0,1)^{K_2}|\boldsymbol 1^\top\boldsymbol\pi^s<1\}$. Notice that, in complete analogy with the general case, if $\boldsymbol\pi^s$ is augmented with $\pi_{1,K_2}$, the components of the resulting vector will sum up to the unity. Hence, the $(r,F_2,\mathbf Z)$ representation of this generic sample selection model would require just to exploit the corresponding function $\boldsymbol r$ as depending on the type of the response $Y_2$. In particular, for a dichotomous response $Y_2$,
\begin{equation}
\boldsymbol r(\boldsymbol\pi^s)=(\pi_{0\cdot},\pi_{1,0})^\top=(\Phi_2(-\eta_1,\infty;\boldsymbol\Omega),\Phi_2(\eta_1,-\eta_2;\boldsymbol\Omega))^\top=\mathcal F(\boldsymbol\eta),\label{eq:GAMdiscr.ssBinary}
\end{equation}
where $\boldsymbol\Omega$ stems from $\boldsymbol\Sigma$ upon imposing the restrictions $\boldsymbol\Gamma=\mathbf L=\mathbf I_2$, and $\Phi_2(-\eta_1,\infty)=\Phi_2(-\eta_1,-\eta_2)+(\Phi_2(-\eta_1,\eta_2)-\Phi_2(-\eta_1,-\eta_2))=\Phi(-\eta_1)$ as from the map characterising the sample selection problem. This bivariate probit model was originally proposed by \cite{Heckman_79} and subsequently extended to encompass penalized regression splines by \cite{MarraRadice_13}. As a natural generalisation, one can also consider the support of $\mathcal K_2$ to be totally ordered, $\#(\mathcal K_2)>2$ like in \cite{MirandaRabeHesketh_06}, whose corresponding representation within our framework comprises
\begin{equation}
\boldsymbol r(\boldsymbol\pi)=\left(\pi_{0\cdot},\pi_{1,1},\ldots,\pi_{1,1}+\cdots+\pi_{1,K_2-1}\right)^\top,\mbox{ and}\nonumber
\end{equation}
\begin{equation}
\mathcal F(\boldsymbol\eta)=\left(\Phi(-\eta_1),r^{-1}(\Phi_2(\eta_{1,1},\eta_{2,1};\boldsymbol\Omega)),\ldots,r^{-1}(\Phi_2(\eta_{1,1},\eta_{2,K_2-1};\boldsymbol\Omega))\right)^\top,\nonumber
\end{equation}
where the generic $r^{-1}(F_2(\eta_{1,k_1},\eta_{2,k_2}))$ can be computed as the non-negative volume of the rectangles $[\eta_{1,k_1-1},\eta_{1,k_1}]\times[\eta_{2,k_2-1},\eta_{2,k_2}]$ in $\mathds R^2$.

\section{Estimation}
Let the conditional distribution of $(\boldsymbol Y|\boldsymbol X=\boldsymbol x)$ obey a Categorical distribution $\mathcal C(\boldsymbol\pi(\boldsymbol x))$ with mass function 
\begin{equation}
f_{\boldsymbol Y|\boldsymbol X}(\boldsymbol y|\boldsymbol x)=\prod_{k\in\mathcal K}\pi_k(\boldsymbol x)^{\mathds1_{\boldsymbol y=k}},\label{eq:GAMdiscr.density}
\end{equation}
where $\mathds 1_{\boldsymbol y=k}$ is a Boolean function that takes value 1 if $(y_1=k_1\wedge\cdots\wedge y_J=k_J)$ and 0 otherwise. Then, after having re-defined the response vector $\bar{\boldsymbol y}=(\mathds1_{\boldsymbol y=1},\ldots,\mathds1_{\boldsymbol y=K})^\top$, the distribution (\ref{eq:GAMdiscr.density}) can be written as
\begin{equation}
f_{\boldsymbol Y|\boldsymbol X}(\bar{\boldsymbol y}|\boldsymbol x)=\textup{exp}\left\{\bar{\boldsymbol y}^\top\boldsymbol\theta-b(\boldsymbol\theta)\right\},\nonumber
\end{equation}
where
\begin{equation}
\theta_k=\textup{ln}\left\{\frac{\pi_k}{1-\sum_k\pi_k}\right\},\mbox{ }\theta_K=0\mbox{ }\mbox{ and }\mbox{ }b(\boldsymbol\theta)=\textup{ln}\left\{1+\sum_k\textup{exp}\{\theta_k\}\right\},\nonumber
\end{equation}
which shows that $\mathcal C(\boldsymbol\pi(\boldsymbol x))$ can be expressed in the exponential form, and hence all the standard properties implied by this family of distributions follow. If we further take $\{(\boldsymbol x_i,\boldsymbol y_i)\}_{i=1}^n$ a sample from $F_{\boldsymbol Y}$ and $F_{\boldsymbol X}$, where the $\boldsymbol y_i$'s are assumed conditionally independent given the regressors, then equation (\ref{eq:GAMdiscr.density}) can also be used to derive the log-likelihood function of any multivariate model for discrete data admitting a $(r,F_J,\mathbf Z)$ form. Specifically, by denoting $\ell_i(\boldsymbol\vartheta)$ the contribution of the $i$-th observation to the log-likelihood, the iterative application of the chain rule results in
\begin{equation}
\nabla_{\boldsymbol\vartheta}\ell_i=\frac{\partial\boldsymbol\eta_k}{\partial\boldsymbol\vartheta}\left(\frac{\partial\mathcal F_k}{\partial\boldsymbol\eta_k}\frac{\partial\pi_k}{\partial\boldsymbol r_k}\frac{\partial\theta_k}{\partial\pi_k}\frac{\partial\ell_i}{\partial\theta_k}\right)=\mathbf D_i^\top\mathbf u_i\mbox{ }\mbox{ }\mbox{ }\mbox{ }\mbox{ }\mbox{and}\mbox{ }\mbox{ }\mbox{ }\mbox{ }\mbox{ }\nabla_{\boldsymbol\vartheta\boldsymbol\vartheta^\top}\ell_i=\mathbf D_i^\top\mathbf W_i\mathbf D_i+\mathbf K_i,\label{eq:GAMdiscr.grHess}
\end{equation}
where
\begin{equation}
\mathbf W_i=\frac{\partial^2\mathcal F_k}{\partial\boldsymbol\eta_k\partial\boldsymbol\eta_k^\top}\frac{\partial\pi_k}{\partial\boldsymbol r_k}\frac{\partial\theta_k}{\partial\pi_k}\frac{\partial\ell_i}{\partial\theta_k}+\frac{\partial\mathcal F_k}{\partial\boldsymbol\eta_k}\frac{\partial\pi_k}{\partial\boldsymbol r_k}\left(\frac{\partial\mathcal F_k}{\partial\boldsymbol\eta_k}\frac{\partial\pi_k}{\partial\boldsymbol r_k}\right)^\top\left\{\frac{\partial^2\theta_k}{\partial\pi_k^2}\frac{\partial\ell_i}{\partial\theta_k}+\left(\frac{\partial\theta_k}{\partial\pi_k}\right)^2\frac{\partial^2\ell_i}{\partial\theta_k^2}\right\}\nonumber
\end{equation}
and
\begin{equation}
\mathbf K_i=\frac{\partial^2\boldsymbol\eta_k}{\partial\boldsymbol\vartheta\partial\boldsymbol\vartheta^\top}\frac{\partial\mathcal F_k}{\partial\boldsymbol\eta_k}\frac{\partial\pi_k}{\partial\boldsymbol r_k}\frac{\partial\theta_k}{\partial\pi_k}\frac{\partial\ell_i}{\partial\theta_k}\nonumber.
\end{equation}
These expressions are analogous to those derived by \cite{Green_84} in the context of iterative re-weighted least squares (IRLS) estimation of likelihood functions. Indeed, the baseline model is rather similar, with the sole relevant difference being the acknowledgment that only in some special cases $\boldsymbol r(\boldsymbol\pi)=\boldsymbol\pi$. In particular, wherever $\boldsymbol r$ is the identity map, $\widetilde{\mathbf u}=\textup{vec}(\mathbf u_1,\ldots,\mathbf u_n)$ reduces to the same simplified expression, $\partial\ell_i/\partial\boldsymbol\eta_k$, that appears in \cite{Green_84}. Factor $\mathbf K_i$ is somehow unusual, and generally it is not reported in the relevant literature on GLMs. In fact, it is structurally equal to $\boldsymbol 0_{p,p}$ wherever each $\boldsymbol\eta_k$ is linear in the parameter vector; however, this may not be true in some instances as shown, for example, in the triangular systems of equations having representation (\ref{eq:GAMdiscr.rFZtriang}). The Information Matrix can also be derived: recalling that, for the exponential family of distributions, $\partial\ell_i/\partial\theta_k=\bar y_k-\pi_k$, $b'(\theta_k)=\pi_k$ and $\partial\theta_k/\partial\pi_k=[b''(\theta_k)]^{-1}=\mathbb V\textup{ar}[\bar y_k]^{-1}=[\pi_k(1-\pi_k)]^{-1}$, we have $\mathbb E[\mathbf K_i]=\boldsymbol 0_{p,p}$, while
\begin{equation}
\mathbb E[\mathbf W_i]=-\frac{\partial\mathcal F_k}{\partial\boldsymbol\eta_k}\frac{\partial\pi_k}{\partial\boldsymbol r_k}\left(\frac{\partial\mathcal F_k}{\partial\boldsymbol\eta_k}\frac{\partial\pi_k}{\partial\boldsymbol r_k}\right)^\top\frac{1}{\pi_k(1-\pi_k)}=:-\overline{\mathbf W}_i,\nonumber
\end{equation}
so that the Fisher Information component corresponding to the $i$-th observation is given by
\begin{equation}
\mathcal I_i=-\mathbb E[\nabla_{\boldsymbol\vartheta\boldsymbol\vartheta^\top}\ell_i]=\mathbf D_i^\top\overline{\mathbf W}_i\mathbf D_i.\nonumber
\end{equation}
Each individual matrix is finally aggregated into appropriate arrays to get a global representation of score and Hessian as follows: $\mathbf D:=(\mathbf D_1^\top|\cdots|\mathbf D_n^\top|\mathbf I_p)^\top$, $\mathbf u:=\textup{vec}(\mathbf u_1,\ldots,\mathbf u_n,\boldsymbol0_p)$, $-\mathbf W:=\textup{diag}(\mathbf W_1,\ldots,\mathbf W_n,\mathbf K)$, with $\mathbf K:=\sum_i\mathbf K_i$, so that $\nabla_{\boldsymbol\vartheta}\ell(\boldsymbol\vartheta)=\mathbf D^\top\mathbf u$, and $\nabla_{\boldsymbol\vartheta\boldsymbol\vartheta^\top}\ell(\boldsymbol\vartheta)=-\mathbf D^\top\mathbf W\mathbf D$.

\subsection{Penalized Likelihood}
The quantities derived above have been obtained only by the knowledge of the $(r,F_J,\mathbf Z)$ representation of the model that, alongside with the panalty matrix $\mathbf S_{\boldsymbol\lambda}$, embodies all the information needed to achieve estimation. Recall that any covariate effect other than a purely parametric specification requires the exploitation of certain features as included in the penalisation term $\mathcal P_{\boldsymbol\lambda}$. To account for them, a Penalized Likelihood (PL) is usually set up for estimation, and the corresponding MPLE is then defined as solution of the following optimisation problem
\begin{equation}
\widehat{\boldsymbol\vartheta}:=\operatornamewithlimits{\arg\max}_{\boldsymbol\vartheta\in\Theta}\ell_{\textup p}(\boldsymbol\vartheta,\boldsymbol\lambda)=\left\{\sum_{i=1}^n\ell_i(\boldsymbol\eta_k(\boldsymbol\vartheta))-\frac{1}{2}\boldsymbol\vartheta^\top\mathbf S_{\boldsymbol\lambda}\boldsymbol\vartheta\right\},\label{eq:GAMdiscr.MPLEdef}
\end{equation}
which is obtained from any fixed value of the smoothing paramter vector $\boldsymbol\lambda$. Because the quadratic form $\mathcal P_{\boldsymbol\lambda}$ is positive semi-definite by construction, the joint estimation of $(\boldsymbol\vartheta,\boldsymbol\lambda)$ would clearly lead to over-fitting since an optimal value for $\ell_{\textup p}(\boldsymbol\vartheta)$ would be reached at a state where $\widehat{\boldsymbol\lambda}=\boldsymbol0$. Our estimation strategy comprises therefore two alternating steps based on the outer iteration scheme originally proposed by \cite{OSYR_86}. Specifically, an estimate $(\widehat{\boldsymbol\vartheta}|\boldsymbol\lambda=\boldsymbol\lambda')$ is first obtained from any value $\boldsymbol\lambda'$ via the maximisation of $\ell_{\textup p}(\boldsymbol\vartheta|\boldsymbol\lambda=\boldsymbol\lambda')$, which is then used to update a value of the tuning parameter vector. The whole procedure is iterated until convergence.

Although a solution to problem (\ref{eq:GAMdiscr.MPLEdef}) can in principle be obtained through any numerical optimisation algorithm, our subsequent analysis requires some of its iterations to be either of Newton-Raphson or of Fisher scoring-type to match with the derivation of the smoothing parameters vector.

\paragraph{P-IRLS Scheme for Estimation}
Rather than handling the log-likelihood maximisation directly, it is convenient to define a penalized iteratively re-weighted least squares (P-IRLS) scheme based on quantities (\ref{eq:GAMdiscr.grHess}). Let us first derive the Taylor series approximation of the function $\nabla_{\boldsymbol\vartheta}\ell_{\textup p,i}$ about the vector $\widehat{\boldsymbol\vartheta}=\boldsymbol\vartheta-\boldsymbol\vartheta_0$, 
\begin{equation}
\nabla_{\widehat{\boldsymbol\vartheta}}\ell_{\textup p,i}\approx\nabla_{\boldsymbol\vartheta_0}\ell_{\textup p,i}+\nabla_{\boldsymbol\vartheta_0\boldsymbol\vartheta_0^\top}\ell_{\textup p,i}(\widehat{\boldsymbol\vartheta}-\boldsymbol\vartheta_0)=\boldsymbol0_p,\nonumber
\end{equation}
where the last equality holds from $\widehat{\boldsymbol\vartheta}$ being the MPLE, and with $\nabla_{\widehat{\boldsymbol\vartheta}}\ell_{\textup p,i}$ standing for $\nabla_{\boldsymbol\vartheta}\ell_{\textup p,i}|_{\boldsymbol\vartheta=\widehat{\boldsymbol\vartheta}}$. Under the assumptions that $\mathbf D$ has full rank $p$, and $\mathbf W$ is positive definite throughout the parameter space $\Theta$, a Newton-Raphson algorithm comprises the non-singular $p\times p$ system of equations for $\boldsymbol\vartheta$
\begin{eqnarray}
(\nabla_{\boldsymbol\vartheta\boldsymbol\vartheta^\top}\ell(\boldsymbol\vartheta^{[\alpha]})-\mathbf S_{\boldsymbol\lambda})\boldsymbol\vartheta^{[\alpha+1]}&=&(\nabla_{\boldsymbol\vartheta\boldsymbol\vartheta^\top}\ell(\boldsymbol\vartheta^{[\alpha]})-\mathbf S_{\boldsymbol\lambda})\boldsymbol\vartheta^{[\alpha]}+(\mathbf S_{\boldsymbol\lambda}\boldsymbol\vartheta^{[\alpha]}-\nabla_{\boldsymbol\vartheta}\ell(\boldsymbol\vartheta^{[\alpha]}))\nonumber\\
(\mathbf D^\top\mathbf W\mathbf D+\mathbf S_{\boldsymbol\lambda})\boldsymbol\vartheta^*&=&\mathbf D^\top\mathbf W\mathbf z,\label{eq:GAMdiscr.IRLS}
\end{eqnarray}
where $\mathbf z:=\mathbf D\boldsymbol\vartheta+\mathbf W^{-1}\mathbf u$ defines the pseudo-data vector associated with any $(r,F_J,\mathbf Z)$ model. Moreover, equation (\ref{eq:GAMdiscr.IRLS}) is expressed in terms of $\boldsymbol\vartheta^*:=\boldsymbol\vartheta^{[\alpha+1]}$, while the dependence of all the other variables on the $[\alpha]$-th iteration is neglected to avoid clutter in the notation. Finally, by noticing that the above system can be recovered directly from the normal equations of a Generalized Least Squares (GLS) regression of $\mathbf z$ onto the columns of $\mathbf D$, using a weight matrix $\mathbf W$ and a ridge-type penalisation, it follows that (\ref{eq:GAMdiscr.IRLS}) corresponds to the closed-form solution of the problem
\begin{equation}
\boldsymbol\vartheta^*=\operatornamewithlimits{\arg\min}_{\boldsymbol t\in\Theta}\big\|\sqrt{\mathbf W}(\mathbf z-\mathbf D\boldsymbol t)\big\|^2+\boldsymbol t^\top\mathbf S_{\boldsymbol\lambda}\boldsymbol t.\nonumber
\end{equation}
In other words, at every $[\alpha]$-th iteration, the GLS recursion produces a closed-form expression to update the optimisation algorithm, and this is repeated until convergence. Apart from giving an elegant solution to the log-likelihood maximisation problem, the P-IRLS algorithm also establishes a correspondence between MPLE and GLS, and this provides us with an equivalent expression smoothing parameter selection can be based on.

\begin{remark}
The use of matrix $\mathbf D$ in the computations above reflects the possibility of dealing with models involving non-linear predictors. In other simpler instances, this quantity reduces to the design matrix $\mathbf Z$, with potential gains in the computational time of the P-IRLS procedure. In fact, $\mathbf D$ would usually depend on some functions of the parameter vector which need to be updated at every iteration; whereas, in the case of $\mathbf D\equiv\mathbf Z$, this quantity can be stored outside the iteration loop.   
\end{remark}

\subsection{Smoothing Parameter Selection}\label{sct:GAMdiscr.smoothsel}      
The correct specification of the ``right'' amount of smoothness is important for any practical modelling in non-parametric regression. In what follows, we adapt the Un-biased Risk Estimator (UBRE; e.g.\ \citealp{Wood_06}) to the present context, so that smoothness selection is achieved from quantities that are directly stemming from the $(r,F_J,\mathbf Z)$ representation of the model; a stable and efficient computational method to implement this criterion is discussed in \cite{Wood_04}.

In principle, vector $\boldsymbol\lambda$ should be estimated in such a way that the fitted curves are as close as possible to the true unknown functions. To this end, let us consider the large sample approximation $-\nabla_{\boldsymbol\vartheta\boldsymbol\vartheta^\top}\ell(\boldsymbol\vartheta)\stackrel{p}{\longrightarrow}\mathcal I$ implied by the likelihood model and, under the regularity conditions listed in Section \ref{sct:GAMdiscr.asymptotics}, it follows $\mathbf D^\top\mathbf W\mathbf D\stackrel{p}{\longrightarrow}\mathbf D^\top\overline{\mathbf W}\mathbf D$, where $\overline{\mathbf W}:=\textup{diag}(\overline{\mathbf W}_1,\ldots,\overline{\mathbf W}_n,\boldsymbol0_{p,p})$. Since $\mathbb E[\mathbf K_i]=\boldsymbol0_{p,p}$, it also holds that $\mathbb E[\mathbf K]=\boldsymbol0_{p,p}$ from the linearity of the expectation, hence $\mathbf K=o_p(1)$ and $\mathbf W=\overline{\mathbf W}+o_p(1)$ as $n\rightarrow\infty$. Further let 
\begin{equation}
\mathbf P=\sqrt{\overline{\mathbf W}}\mathbf D(\mathbf D^\top\overline{\mathbf W}\mathbf D+\mathbf S_{\boldsymbol\lambda})^{-1}\mathbf D^\top\sqrt{\overline{\mathbf W}}\nonumber
\end{equation}
denote the influence matrix of the associated GLS model, namely the array such that the predicted values of the response $\sqrt{\overline{\mathbf W}}\overline{\mathbf z}$ can be written as $\widehat{\boldsymbol\mu}:=\sqrt{\overline{\mathbf W}}\mathbf D\boldsymbol\vartheta^*=\mathbf P\sqrt{\overline{\mathbf W}}\overline{\mathbf z}$, and $\overline{\mathbf z}:=\mathbf z(\overline{\mathbf W})$ be the pseudo-data vector evaluated at the asymptotic weight matrix. Then, by letting $\boldsymbol\mu:=\mathbb E[\sqrt{\overline{\mathbf W}}\overline{\mathbf z}]=\sqrt{\overline{\mathbf W}}\mathbf D\boldsymbol\vartheta$ be the expected value of the GLS response, we define $\widehat{\boldsymbol\lambda}$ as the minimiser of the expected Mean Squared Error (MSE) of $\widehat{\boldsymbol\mu}$. Namely 
\begin{eqnarray}
\widetilde n^{-1}\mathbb E\|\boldsymbol\mu-\widehat{\boldsymbol\mu}\|^2&=&\widetilde n^{-1} \mathbb E\big\|\sqrt{\overline{\mathbf W}}\overline{\mathbf z}-\mathbf P\sqrt{\overline{\mathbf W}}\overline{\mathbf z}-\boldsymbol\varepsilon\big\|^2\nonumber\\
&=&\widetilde n^{-1}\mathbb E\left[\big\|\sqrt{\overline{\mathbf W}}(\overline{\mathbf z}-\mathbf D\boldsymbol\vartheta^*)\big\|^2+\|\boldsymbol\varepsilon\|^2-2\big\langle\sqrt{\overline{\mathbf W}}\overline{\mathbf z}-\mathbf P\sqrt{\overline{\mathbf W}}\overline{\mathbf z};\boldsymbol\varepsilon\big\rangle\right],\nonumber
\end{eqnarray}
where the stochastic term above is given by $\overline{\mathbf W}^{-1/2}\mathbf u=:\boldsymbol\varepsilon\sim(\boldsymbol0_{\widetilde n},\mathbf I_{\widetilde n})$ since $\mathbb E[\mathbf u]=\boldsymbol0_p$ and $\mathbb{V}\textup{ar}[\mathbf u]=\overline{\mathbf W}$, while the expectation of the inner product results in
\begin{eqnarray}
-2\widetilde n^{-1}\mathbb E\big\langle\sqrt{\overline{\mathbf W}}\overline{\mathbf z}-\mathbf P\sqrt{\overline{\mathbf W}}\overline{\mathbf z};\boldsymbol\varepsilon\big\rangle&=&-2\widetilde n^{-1}\mathbb E\left[\boldsymbol\varepsilon^\top(\mathbf I_{\widetilde n}-\mathbf P)(\sqrt{\overline{\mathbf W}}\mathbf D\boldsymbol\vartheta+\boldsymbol\varepsilon)\right]\nonumber\\
&=&-2+2\widetilde n^{-1}\mathbb E[\boldsymbol\varepsilon^\top\mathbf P\boldsymbol\varepsilon]=-2+2\widetilde n^{-1}\textup{tr}(\mathbf P).\nonumber
\end{eqnarray}
Then, the corresponding UBRE criterion for the $[\alpha+1]$-th iteration step reads as
\begin{equation}
\boldsymbol\lambda^{[\alpha+1]}:=\operatornamewithlimits{\arg\min}_{\boldsymbol\lambda}\mathcal V_u(\boldsymbol\lambda):=\big\|\sqrt{\overline{\mathbf W}}^{[\alpha+1]}(\overline{\mathbf z}^{[\alpha+1]}-\mathbf D^{[\alpha+1]}\boldsymbol\vartheta^{[\alpha+1]})\big\|^2/\widetilde n-1+2\kappa\textup{tr}(\mathbf P^{[\alpha+1]})/\widetilde n,\label{eq:GAMdiscr.UBRE}
\end{equation}
where $\widetilde n$ is a given multiple of the sample size as determined by the dimension of $\boldsymbol\eta_k$, and accounts for the multivariate nature of the framework. For instance, if we consider a bivariate model where both the responses are ordinal polychotomous with just one association paramter $\gamma$, it follows that, for every individual $i$, the corresponding array
\begin{equation}
\boldsymbol\eta_k=(\eta_{k_1-1,k_2-1},\eta_{k_1-1,k_2},\eta_{k_1,k_2-1},\eta_{k_1,k_2},\gamma)^\top\nonumber
\end{equation}
is 5-dimensional and $\widetilde n=5n+p$. An additional inflation parameter $\kappa$ has been included in the UBRE criterion and it can be increased from its usual value of 1 in order to obtain smoother models. In effect, based on experimental results, \cite{KimGu_04} suggested to locate $\kappa\in[1.2,1.4]$ to correct the tendency of prediction error criteria to over-fit the estimated curves. Notice that the trace of the influence matrix $\mathbf P$ represents the effective degrees of freedom of the model; in a penalised framework, they usually differ from the number of parametric model components since the presence of the penalisation in the fitting algorithm can suppress some dimensions of the parameter space. As a final remark, expression (\ref{eq:GAMdiscr.UBRE}) can also be interpreted in term of the log-likelihood AIC:

\begin{prop}\label{prop:GAMdiscr.AIC}
Let $\ell(\boldsymbol\vartheta)$ be the log-likelihood function of a model admitting a ridge-type penalized GLM form $(r,F_J,\mathbf Z)$ and penalty matrix $\mathbf S_{\boldsymbol\lambda}$, then the UBRE expression (\ref{eq:GAMdiscr.UBRE}), $\mathcal V_u(\boldsymbol\lambda)$, is proportional to the Akaike Information Criterion (AIC) with the parameter space dimensionality corrected for the presence of $\mathbf S_{\boldsymbol\lambda}$, namely
\begin{equation}
\mathcal V_u(\boldsymbol\lambda)\propto2\textup{tr}(\mathbf P)-2\ell(\boldsymbol\vartheta^*).\nonumber
\end{equation}
\end{prop}

\begin{proof}
See Supplementary Material \ref{sctSM:GAMdiscr.AICproof}.
\end{proof}

\begin{algorithm}[t]
\caption{Pseudo-code for MPLE Computation under a Outer Iteration Scheme} 
\label{alg:GAMdiscr.MPLE}      
\begin{algorithmic} 
\REQUIRE $\alpha\in(0,\texttt{iter.max})$; $\kappa\ge1$; $(r,F_J,\mathbf Z)$; $\mathbf S_{\boldsymbol\lambda}$
\STATE $\boldsymbol\vartheta^{[0]}$, $\boldsymbol\lambda^{[0]}$
\WHILE{$\alpha\le\texttt{iter.max}\mbox{ }\vee\mbox{ }\max\big|\boldsymbol\vartheta^{[\alpha+1]}-\boldsymbol\vartheta^{[\alpha]}\big|\ge10^{-6}$}
\STATE{$\boldsymbol\vartheta^{[\alpha+1]}\Leftarrow\max_{\boldsymbol\vartheta}\left[\ell(\boldsymbol\vartheta^{[\alpha]})-\boldsymbol\vartheta^{[\alpha]\top}\mathbf S_{\boldsymbol\lambda^{[\alpha]}}\boldsymbol\vartheta^{[\alpha]}\right]$}\vspace{0.1cm}
\STATE{$\mathbf D_i^{[\alpha+1]}\Leftarrow\left.\frac{\partial\boldsymbol\eta_k}{\partial\boldsymbol\vartheta}\right|_{\boldsymbol\vartheta=\boldsymbol\vartheta^{[\alpha+1]}}$; $\mathbf u_i^{[\alpha+1]}\Leftarrow\left.\frac{1}{\pi_k}\frac{\partial\mathcal F_k}{\partial\boldsymbol\eta_k}\frac{\partial\pi_k}{\partial\boldsymbol r_k}\right|_{\boldsymbol\vartheta=\boldsymbol\vartheta^{[\alpha+1]}}$}\vspace{0.1cm}
\STATE{$\mathbf W_i^{[\alpha+1]}\Leftarrow\left.\frac{1}{\pi_k}\frac{\partial^2\mathcal F_k}{\partial\boldsymbol\eta_k\partial\boldsymbol\eta_k^\top}-\mathbf u_i^{[\alpha+1]}\mathbf u_i^{[\alpha+1]\top}\right|_{\boldsymbol\vartheta=\boldsymbol\vartheta^{[\alpha+1]}}$; $\mathbf K_i^{[\alpha+1]}\Leftarrow\left.\frac{\partial^2\boldsymbol\eta_k}{\partial\boldsymbol\vartheta\partial\boldsymbol\vartheta^\top}\mathbf u_i^{[\alpha+1]}\right|_{\boldsymbol\vartheta=\boldsymbol\vartheta^{[\alpha+1]}}$}\vspace{0.1cm}
\STATE{compute $\mathbf D^{[\alpha+1]}$, $\mathbf u^{[\alpha+1]}$, $\mathbf K^{[\alpha+1]}$, $\mathbf W^{[\alpha+1]}$, $\mathbf z^{[\alpha+1]}$ and $\mathbf P^{[\alpha+1]}$}\vspace{0.1cm}
\STATE{$\boldsymbol\lambda^{[\alpha+1]}\Leftarrow\min_{\boldsymbol\lambda}\left[\big\|\sqrt{\mathbf W}^{[\alpha+1]}(\mathbf z^{[\alpha+1]}-\mathbf D^{[\alpha+1]}\boldsymbol\vartheta^{[\alpha+1]})\big\|^2/\widetilde n-1+2\kappa\textup{tr}(\mathbf P^{[\alpha+1]})/\widetilde n\right]$}
\ENDWHILE
\end{algorithmic}
\end{algorithm}

\afterpage{
\begin{figure}[t]\centering
\includegraphics[width=0.32\textwidth,angle=270]{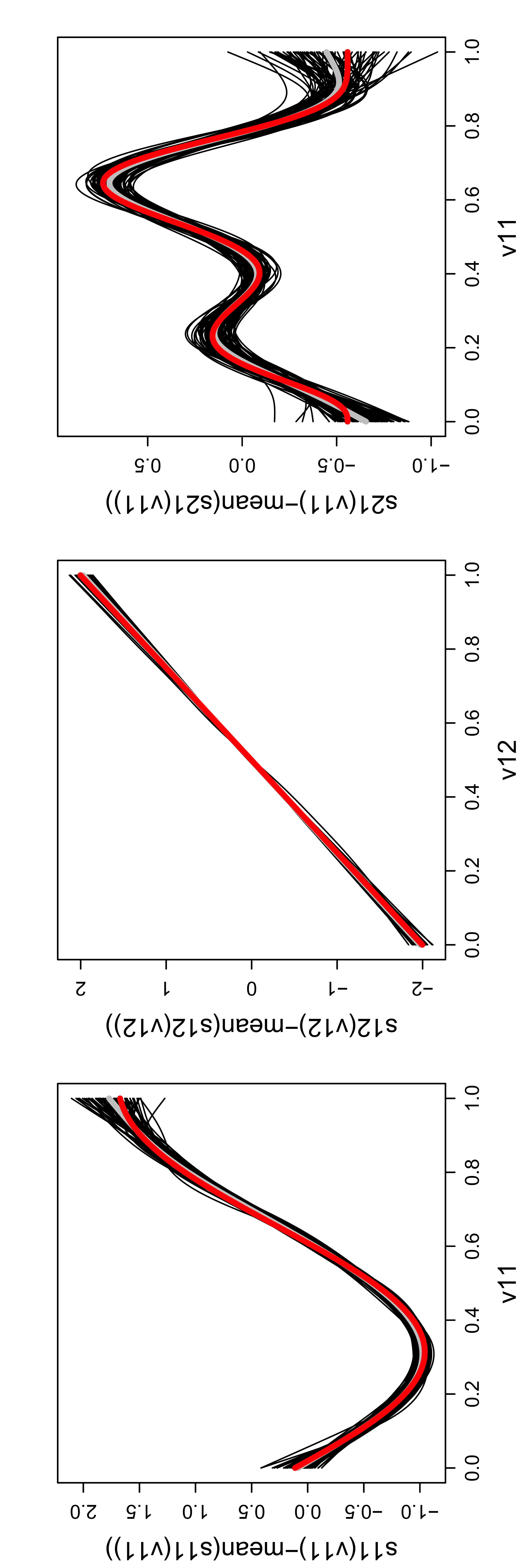}
\caption{Monte Carlo simulation based on 100 replications for the non-parametric components of the bivariate triangular ordered probit model described by form (\ref{eq:GAMdiscr.rFZtriang}): the first two graphs refer to the smooths included in the first equation, while the last corresponds to the model for $Y_2$. The true functions are depicted in red and an average estimated curve is represented in gray (notice that for many observations it overlaps with the corresponding truth value). The computations were performed by the \texttt R function \texttt{SemiParCLM} that implements (\ref{eq:GAMdiscr.rFZtriang}) and some nested models following Algorithm \ref{alg:GAMdiscr.MPLE}, on a model comprising 10,000 observation and correlation coefficient $\rho=0.9$. Details on the Data Generating Process are given in Supplementary Material \ref{appx:GAMdiscr.DGP}.}
\label{fig:GAMdiscr.simSmooths}
\end{figure}
}

\noindent The structure of the resulting fitting procedure is detailed in Algorithm \ref{alg:GAMdiscr.MPLE}, and an illustration of its ability to recover the unknown smooth functions shown in Figure \ref{fig:GAMdiscr.simSmooths} by a simulation study for the triangular model (\ref{eq:GAMdiscr.rFZtriang}). Some considerations on the asymptotic behaviour of the proposed estimator as well as a method to compute confidence intervals for the included smooths are detailed in Supplementary Materials \ref{sct:GAMdiscr.asymptotics} and \ref{sctSM:GAMdiscr.CI}, respectively.

\section{Real Data Illustration: HIV Prevalence in Zambia}
We illustrate now the proposed framework via the estimation of a sample selection instance. In doing this, we specialise our structure to describe a bivariate probit regression with association parameter explained through an additive linear predictor (\citealp{RMW_15}). This feature is attractive in the context of unmeasured confounding as it allows to account for various degrees of non-random sample selection across observations, and it helps to explain how the association between the relevant outcomes is affected by common unobservables for different individuals and covariates.

The resulting model is then applied to data from the 2007 Zambia Demographic Health Survey (DHS) to flexibly estimate the prevalence of HIV in the Zambian male population. Our analysis complements the study of \cite{MBMR_15} through the inclusion of non-parametric covariate effects, and the specification of the aforementioned elements proper of a distributional regression. The following discussion is further extended by \cite{MRBWM_15}, to which we refer the reader for more extensive and thoughtful argumentations. All the relevant computations presented in the study are performed in the \texttt R environment (\citealp{R_15}) using the package \texttt{SemiParBIVProbit} (\citealp{MarraRadice_15}) which implements the ideas discussed in this article for the binary case, and whose corresponding representation in the $(r,F_2,\mathbf Z)$ form has been previously given in (\ref{eq:GAMdiscr.ssBinary}). Notice, however, that because of existing restrictions in the original data-set diffusion, just a one generated from it can be used for reproducibility purposes, and it can be accessed from the above-mentioned package (\texttt{hiv}); details on the employed model specification can be found in \cite{MBMR_15}.

\subsection{A Dichotomous Regression Defined Through Bivariate Copulae}
The models presented in Section \ref{sct:GAMdiscr.ConfoundingSS} where originally defined through a bivariate Gaussian distribution. This may be a strong assumption though, especially in applied disciplines where symmetries are unlikely or implausible: a mitigation of these constraints can then be achieved by extending the framework to copulae. As a first definition, let $F_{1,j}$ be the marginal distribution of the $j$-th component of $\boldsymbol Y^*$, and consider the map $\mathcal C_J:[0,1]^{J}\longrightarrow[0,1]$, such that
\begin{equation}
\mathcal C_J(F_{1,1}(Y_1^*),\ldots,F_{1,J}(Y_J^*))=:\mathcal C_J(U_1,\ldots,U_J)\nonumber
\end{equation}
is the joint cumulative distribution function of $(U_1,\ldots,U_J)^\top$. Then $U_j$ is uniformly distributed for each $j\in\mathcal J$, and $\mathcal C_J$ is called the $J$-variate copula of the vector $(Y_1^*,\ldots,Y_J^*)^\top$ which is a \textit{bona fide} multivariate distribution function under the Sklar's Theorem (\citealp{Sklar_59}). Notice that, simply by denoting $F_J(\boldsymbol\eta_k)\equiv\mathcal C_J(F_{1,1}(\eta_{1,k_1}),\ldots,F_{1,J}(\eta_{J,k_J}))$, any copula representation in principle belongs already to the class of models we have introduced; for a full account of copulae and their properties we refer to the monograph of \cite{Nelsen_06}.

A bivariate copula regression for dichotomous responses sets the probability of any $\pi_k\in\boldsymbol\pi$, for $k\in\{(0,0),(0,1),(1,0),(1,1)\}=:\mathcal K$, as
\begin{eqnarray}
\pi_k&:=&\mathbb P[Y_1=k_1,Y_2=k_2]\nonumber\\
&=&(r^{-1}\circ\mathcal C_2)(F_{1,1}(\eta_{1,k_1}),F_{1,2}(\eta_{2,k_2});\gamma^*)=\mathcal C_2(F_{1,1}(\eta_{1,k_1}),F_{1,2}(\eta_{2,k_2});\gamma),\nonumber
\end{eqnarray}
where the last equality follows for $r$ being the identity map, and with $\gamma$ being an association parameter measuring the dependence between the two marginals. For optimisation purposes it is sometimes desirable to unbound the support of $\gamma$, hence a specific copula-dependent transformation $\gamma^*$ may be applied, which is taken here to be a function of the covariate vector $\mathbf x_3$. Since the corresponding $(r,F_2,\mathbf Z)$ representation of this model for non-random sample selection is given by
\begin{equation}
\left(\pi_k^s,\mathcal C_2(\Phi(\eta_{1,k_1}),\Phi(\eta_{2,k_2});\gamma^*(\mathbf x_3)),\mathbf Z\right),\nonumber
\end{equation}
in the proceeding all the IRLS quantities needed to perform estimation and inference can be derived from this, while the binding copula is intentionally left unspecified.

The specialisation of the model for dichotomous responses simplifies the generic framework considerably. In particular, by neglecting any triangular structure ($\boldsymbol\Gamma=\mathbf L=\mathbf I_3$), $\mathbf D_i$ reduces to the $3\times p$ design matrix $\mathbf X_i=(\mathbf x_{1,i},\mathbf x_{2,i},\mathbf x_{3,i})^\top$, and $\mathbf K_i=\boldsymbol 0_{p,p}$; whereas the GLM representation implies $\partial \ell_i/\partial\theta_k=1-\pi_k$ for any $k\in\mathcal K$ actually observed, $\partial^2\theta_k/\partial\pi_k^2=\pi_k^{-1}(1-\pi_k)^{-2}-\pi_k^{-2}(1-\pi_k)^{-1}$ and $\partial^2\ell_i/\partial\theta_k^2=-b''(\theta_k)=-\pi_k(1-\pi_k)$. Let now $\boldsymbol\eta_k=(\eta_{1,i},\eta_{2,i},\gamma_i^*)^\top$ be the vector of the linear predictors evaluated at the $i$-th individual, where the subscript for $\gamma_i^*$ is introduced to remark its dependence on $\mathbf x_{3,i}$, then we can further decompose
\begin{equation}
\frac{\partial\mathcal F_k}{\partial\boldsymbol\eta_k}=\frac{\partial\boldsymbol F_k}{\partial\boldsymbol\eta_k}\frac{\partial\mathcal C_k}{\partial\boldsymbol F_k}\label{eq:GAMdiscr.copulaDec}
\end{equation}
to make explicit the contribution of the marginal distributions to $\mathcal F_k$. The score and the main component of Hessian matrix are then
\begin{equation}
\nabla_{\boldsymbol\vartheta}\ell_i=\frac{1}{\pi_k}\mathbf X_i^\top\frac{\partial\boldsymbol F_k}{\partial\boldsymbol\eta_k}\frac{\partial\mathcal C_k}{\partial\boldsymbol F_k}\mbox{ }\mbox{ }\mbox{ }\mbox{ }\mbox{ }\mbox{and}\mbox{ }\mbox{ }\mbox{ }\mbox{ }\mbox{ }\mathbf W_i=\frac{1}{\pi_k}\frac{\partial^2\mathcal F_k}{\partial\boldsymbol\eta_k\partial\boldsymbol\eta_k^\top}-\frac{1}{\pi_k^2}\frac{\partial\boldsymbol F_k}{\partial\boldsymbol\eta_k}\frac{\partial\mathcal C_k}{\partial\boldsymbol F_k}\left(\frac{\partial\boldsymbol F_k}{\partial\boldsymbol\eta_k}\frac{\partial\mathcal C_k}{\partial\boldsymbol F_k}\right)^\top\nonumber
\end{equation}
with
\begin{equation}
\frac{\partial^2\mathcal F_k}{\partial\boldsymbol\eta_k\partial\boldsymbol\eta_k^\top}=\frac{\partial\boldsymbol F_k}{\partial\boldsymbol\eta_k}\frac{\partial^2\mathcal C_k}{\partial\boldsymbol F_k\partial\boldsymbol F_k^\top}\left(\frac{\partial\boldsymbol F_k}{\partial\boldsymbol\eta_k}\right)^\top+\frac{\partial^2\boldsymbol F_k}{\partial\boldsymbol\eta_k\partial\boldsymbol\eta_k^\top}\frac{\partial\mathcal C_k}{\partial\boldsymbol F_k}.\nonumber
\end{equation}
If we further assume Standard Normal marginals for both the components of $\boldsymbol Y$, (\ref{eq:GAMdiscr.copulaDec}) specialises as
\begin{equation}
\frac{\partial\mathcal F_k}{\partial\boldsymbol\eta_k}=\left[\begin{array}{ccc}\frac{\partial\Phi(\eta_{1,i})}{\partial\eta_{1,i}} & 0 & 0\\
0 & \frac{\partial\Phi(\eta_{2,i})}{\partial\eta_{2,i}} & 0\\
0 & 0 & \frac{\partial\gamma_i}{\partial\gamma_i^*}\end{array}\right]
\left[\begin{array}{c}\frac{\partial\mathcal C_k(\cdot,\cdot;\cdot)}{\partial\Phi(\eta_{1,i})}\\\frac{\partial\mathcal C_k(\cdot,\cdot;\cdot)}{\partial\Phi(\eta_{2,i})}\\\frac{\partial\mathcal C_k(\cdot,\cdot;\cdot)}{\partial\gamma_i}\end{array}\right]\nonumber\\
=\left[\begin{array}{c}\phi(\eta_{1,i})h_{1,i}\\\phi(\eta_{2,i})h_{2,i}\\\frac{\partial\gamma_i}{\partial\gamma_i^*}h_{3,i}\end{array}\right],\nonumber
\end{equation}
where $h_{j,i}$ and $\partial\gamma_i/\partial\gamma_i^*$ are quantities specific to the copula employed. Moreover, $\partial^2\mathcal C_k/\partial\boldsymbol F_k\partial\boldsymbol F_k^\top$ is the symmetric matrix with generic element $h_{l,m,i}=\partial h_{l,i}/\partial\Phi(\eta_{m,i})$, $l,m=1,\ldots,3$, under the notational abuse $\eta_{3,i}:=\gamma_i$, and
\begin{eqnarray}
\frac{\partial^2\boldsymbol F_k}{\partial\boldsymbol\eta_k\partial\boldsymbol\eta_k^\top}\frac{\partial\mathcal C_k}{\partial\boldsymbol F_k}=\textup{diag}\left(\begin{array}{ccc}-h_{1,i}\phi(\eta_{1,i})\eta_{1,i} & -h_{2,i}\phi(\eta_{2,i})\eta_{2,i} & h_{3,i}\frac{\partial^2\gamma_i}{\partial\gamma_i^{*2}}\end{array}\right).\nonumber
\end{eqnarray}
The derivations above make it clear that $\mathbf W_i$ is a symmetric $3$-dimensional matrix whose generic element $w_{l,m,i}$, for $l,m\neq3$, after some tedious algebra, is given by
\begin{equation}
w_{l,m,i}=\frac{1}{\pi_k}\left[h_{l,m,i}\phi(\eta_{l,i})\phi(\eta_{m,i})-\mathds 1_{l=m}h_{l,i}\phi(\eta_{l,i})\eta_{l,i}\right]-\frac{1}{\pi_k^2}h_{l,i}h_{m,i}\phi(\eta_{l,i})\phi(\eta_{m,i})\mbox{ }\mbox{ }\mbox{ }l,m\neq3,\nonumber
\end{equation}
and expressions for $l,m=3$ can be obtained in a similar fashion as based on the quantities derived above. Finally, the $i$-th addendum defining the Hessian matrix is simply $\mathcal H_i=\mathbf X_i^\top\mathbf W_i\mathbf X_i$, and the pseudo-data vector $\mathbf z_i=\mathbf X_i\boldsymbol\vartheta-\mathbf W_i^{-1}\mathbf u_i$ is $3\times1$, with $\mathbf u_i=\pi_k^{-1}(\partial\mathcal F_k/\partial\boldsymbol\eta_k)$.

\subsection{Background and Results}\label{sct:Zambia}
HIV prevalence in a population is defined as the fraction of people who are infected or, expressed equivalently, as the probability that a randomly drawn individual has the disease. Accurate estimation of the HIV prevalence is essential to policy makers to develop effective control programmes and interventions. Only in recent years, however, in countries where the diffusion of the virus is generalised epidemic, the lack of available administrative data has been overcome by the intensive use of population-based surveys (\citealp{BGW_03}). This is an important new source of data: prior to their introduction, national estimates have prevalently relied on some number of sentinel antenatal clinics  (\citealp{UNAIDS.WHO_03}), whose data may nonetheless present different sources of bias. First of all, their samples are based only on sexually active women who are pregnant and attend a clinic; secondly the location of the facilities, mostly concentrated in urban areas, may also induce a bias even in the subpopulation of pregnant women. These points have been elucidated and discussed with greater details in \cite{MMH_08} and \cite{ADCP_14}, among the others.

On the other hand, participation rates for HIV testing in national surveys are generally low, and ranges from 72\% for men to 77\% for women in the 2007 Zambia DHS (\citealp{HSCHZB_12}), although even lower peaks are recorded in the 2004 Malawi DHS (63\% and 70\%, respectively). There are potentially many reasons inducing this pattern, including concerns, lack of incentive to participate, survey fatigue or migration of those targeted for interview (\citealp{Gersovitz_11}; \citealp{Sterck_13}; \citealp{MBMR_15}); missing data on respondents' HIV status represent therefore a not necessarily less severe cause of bias than the ones already mentioned above. This case study focuses on refusal to be tested for HIV, which is commonly regarded as the main reason of missingness in surveys.

Notice, however, that the use of imputation or weighting techniques are likely to produce biased estimates if the selection mechanism does not occur at random, an assumption violated wherever the reasons of the refusal to test are caused by some unobserved factors. This is the case, for example, when individuals refuse to screen because they already know (or correctly predict) their HIV status, and fear others will learn about their seropositivity if they participate in the survey (\citealp{MBMR_15}). The framework introduced in this article allows us to estimate a Heckman-type selection model which is able to account for the possibility that data are missing not at random. Specifically, this is achieved by modelling item non-response as a function of unobserved variables that also affect the individual HIV status, and by specifying the selection mechanism together with an assumption on the distribution of the unobservables. To foster the identification of the causal mechanism in the study an exclusion restriction is imposed: that is we qualify the dependence of the missing data mechanism on a relevant variable independent of both the outcome of interest given the willingness to take the test, and the unobservables. This is usually labelled an instrument in econometrics and epidemiology, and the interviewer identity is regarded here as a valid instrument to be employed. In fact, previous researches, including \cite{BBWKC_11}, \cite{HSCHZB_12}, \cite{JGWT_14} and \cite{MBMR_15}, have successfully included such a variable in their studies, on the grounds that interviewer identity generally predicts consent to be tested, but it is unlikely it also affects the actual HIV status.

A pictorial representation of the effects on the estimates of applying a sample selection model is reported in Figure \ref{fig:GAMdiscr.maps}. By comparison with the first map, the second one shows immediately how the simple imputation of the values under a random missingness assumption may severely underestimate the HIV prevalence in the Zambian provinces. The imputation has been conducted by making predictions from the univariate model upon discarding the missing values. The third map depicts instead how the association parameter of the employed copula varies among the different regions of the country, and it has been constructed by exploiting its dependence on the geographical location where the survey took place.

\begin{figure}[t]\centering
\includegraphics[width=0.30\textwidth, angle=270]{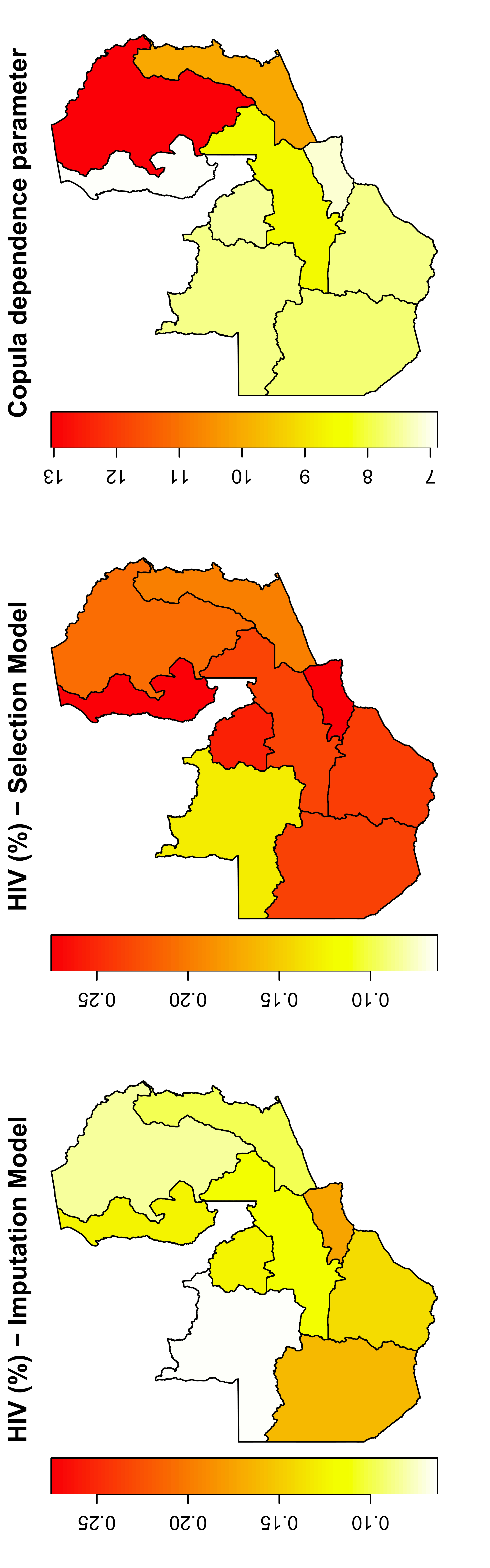}
\caption{First two panels: HIV prevalence for the male population in nine of the ten Provinces of Zambia (Northern, Muchinga, as well as part of Eastern have been merged because of the data availability) applying an imputation model not accounting for the possible presence of values missing not at random, and the corresponding estimates when a bivariate model is fitted instead, respectively. Third panel: the estimated absolute values of the association parameter, with range $(1,\infty)$, in a Joe copula rotated counterclockwise of $90^\circ$. The higher its value, the stronger the estimated association between the two equations; that is, the more relevant the impact of neglecting unobservables in the estimation of the HIV prevalence. The spatial effects are obtained here by specifying appropriately the penalty matrix as described in Section \ref{sct:GAMdiscr.GLMpenal}.} 
\label{fig:GAMdiscr.maps}
\end{figure}

Figure \ref{fig:GAMdiscr.smooths} then reports the smooth function estimates for the treatment and outcome equations, along with their different degrees of non-linearities and associated point-wise confidence intervals, when a Joe$_{90}$ copula model is fitted to the Zambia DHS data; the subscript is used to denote the corresponding copula's degrees of rotation. Notice that, compared to a bivariate Gaussian, the Joe copula is characterised by having a stronger dependence in one tail of the distribution, and its choice for our study has been motivated by the implied negative association between the two marginals, as we would expect to occur wherever persons refuse to be tested on the basis of some prior knowledge of their HIV+ status. Other existing competitors allowing for the same sign of association include models based on the bivariate Gaussian, Frank, Clayton$_{90;270}$, Joe$_{270}$ and Gumbel$_{90;270}$ copulae, which are all implemented in \texttt{SemiParBIVProbit} and discussed within a system of equations in \cite{RMW_15}. As based on information criteria, we found that the Joe$_{90}$ is the best fitting to the male population data, hence our decision to report only selected estimates obtained from this distribution.

\begin{figure}[t]\centering
\includegraphics[width=0.6\textwidth, angle=270]{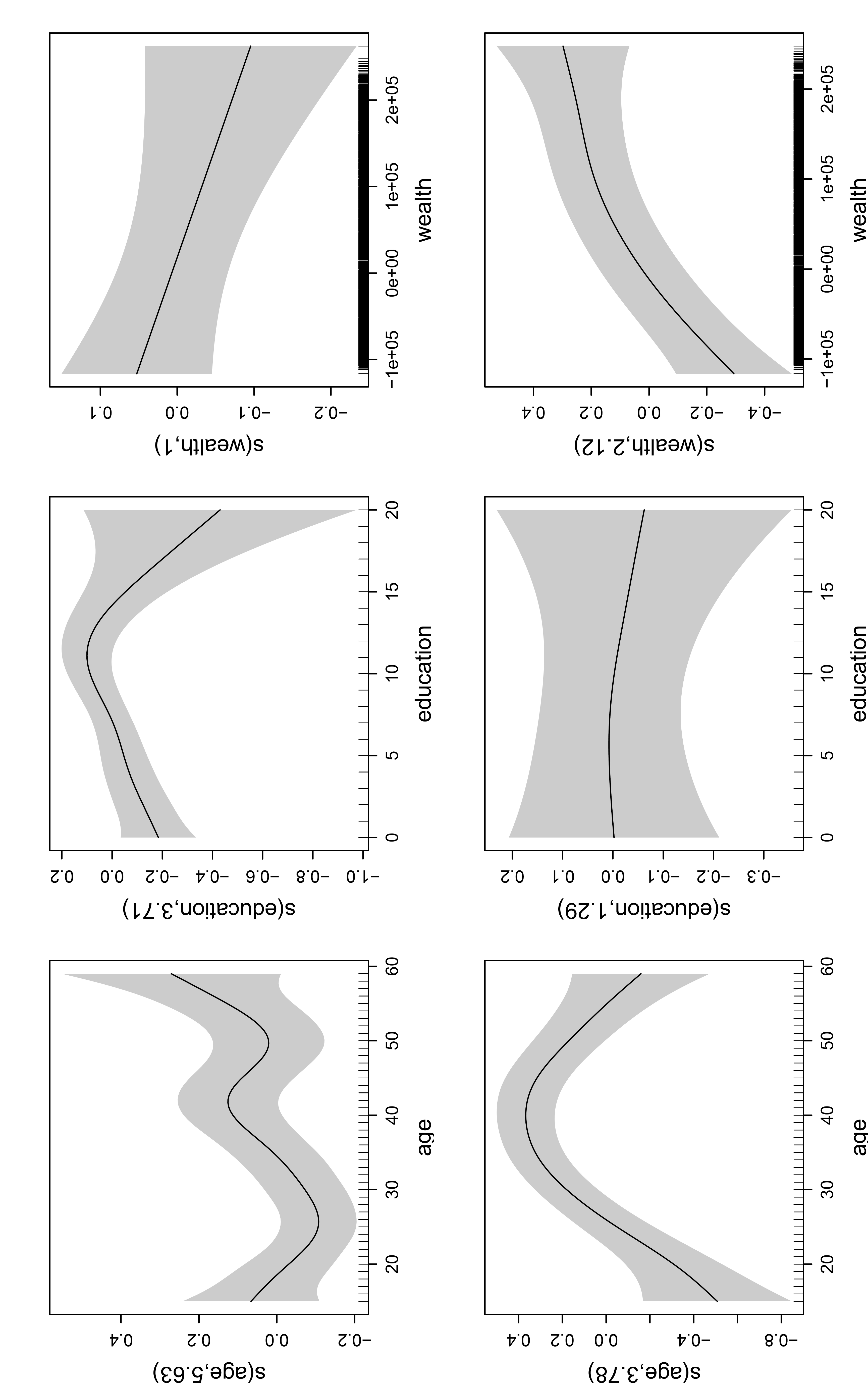}
\caption{Top panel: smooth function estimates and associated 95\% point-wise confidence intervals in the treatment equation obtained by applying the Joe$_{90}$ regression spline model on the 2007 Zambia DHS data. The method employed to derive the confidence bands is described in Supplementary Material \ref{sctSM:GAMdiscr.CI}. Results are plotted on the scale of the linear predictor and the jittered rug plot, at the bottom of each graph, shows the covariate values.The smooth components are represented using low-rank penalized thin plate regression splines (\citealp{Wood_03}) with basis dimensions equal to 10 and penalties based on second order derivatives. The numbers in brackets in the y-axis captions are the effective degrees of freedom of the smooth curves. Bottom panel: estimated smooth functions in the outcome equation.}
\label{fig:GAMdiscr.smooths}
\end{figure}

As a final remark, we shall stress that the assumption of distinct distributions may in principle lead to different estimates of the HIV prevalence (although it seems not to be an issue in this particular application as reported, for instance, by \citealp{MBMR_15}), as well as these can be impacted by the specific functional form of the covariates employed. To deal with this critic, some authors advanced instead the identification of a \textit{region} (rather than of a singleton) of plausible values in which the parameters of interest necessarily lie, given the available data and the maintained assumptions. This switch from \textit{point} to \textit{partial} identification is discussed in general terms in \cite{Manski_95,Manski_03} and \cite{HorowitzManski_00}, and applied to a similar HIV context by \cite{ADCP_14}. Although theoretically valid and appealing, a major drawback of this approach is the realistic possibility of obtaining large width of the estimated bounds: this in turn may let the communication of any result to policymakers harsh even in the case where the identifying region is shrunken by the imposition of a monotone instrumental variable. 

Acknowledging this issue, our estimated model extends the traditional Heckman-type by accounting for three degrees of flexibility: the inclusion of non-parametric effects in the representation of the covariate-response functional form, the specification of bivariate copulae to detect more complex dependence structures than classical distributions usually assume, and the direct modelling of the association parameters in terms of some predictors. It is our hope, in this way, to conjugate both the point and partial identification strengths by providing the researchers with a set of flexible tools aimed at exploring the identifying region widely, and so to make better informed judgments about the robustness of their results wherever a point estimate is sought.

\section{Discussion}
This paper has devised a generic framework for the representation and estimation of a Generalized Linear Model for a $J$-dimensional vector of discrete responses, with a ridge-type penalisation term employed in the fitted algorithm. The resulting class of models allows us to include non-parametric and spatial covariate effects, among the others, as represented through the penalty matrix $\mathbf S_{\boldsymbol\lambda}$. In this way, a baseline multivariate Generalized Additive Model has been effectively extended to encompass different kind of modelling instances within the same unifying framework. In fact, by translating the approach of \cite{PTG_14} to the multivariate case, only the $(r,F_J,\mathbf Z)$ form and the matrix $\mathbf S_{\boldsymbol\lambda}$ are formally required to apply the proposed estimation algorithm and related inferential results to different models in the class.

In particular, once the class has been described in some generality, we have introduced a number of bivariate models employed in the literature to account for the possible presence of residual confounding in observational studies. The proposed representation provided us with a flexible machinery able to extend these models in various directions, foremost towards the additive semi-parametric specification of the linear predictors in the spirit of (V)GAMs. This is, \textit{per se}, already a relevant issue in applied research, since it permits to handle a data-driven representation of the covariate-response relationship and hence to alleviate a possible source of bias stemming from model mis-specification. Moreover, we have described how the framework can be further specified in order to include multivariate distributions as computed by copulae of univariate marginals: some analytical results have been derived for Normal marginals within a bivariate dichotomous regression model.

A further feature illustrated by the article has been the direct modelling of any copula association parameter in terms of known predictors. As shown in the analysis of non-random sample selection for the 2007 DHS Zambia dataset, this characteristic is attractive as it allowed to quantify the strength of the unobservables within the different provinces of the country, and this in turns enabled us to provide new insights about the severity of the non-response issue in the study. In particular, Figure \ref{fig:GAMdiscr.maps} showed that the magnitude of the copula association parameter can vary considerably even between geographically close provinces, like Northern and Luapula in the example. On this point, the relevant literature has already stressed that demographic and environmental factors, like the presence of cities or high density housing, may impact the estimates of the HIV prevalence. Hence, the combination of this knowledge with the possibility of letting the association parameters depend on observed variables seems to us an attractive feature that could be investigated more closely. 

As a natural specification of the proposed framework, the practical implementation of models involving ordinal responses are being developed, whereas the estimation of higher dimensional systems of equations is still limited by the necessity of computing multivariate integrals with a good degree of accuracy. Under this respect, the exploitation of a more comprehensive class of models for copula distributions may be beneficial, possibly by allowing the non-parametric estimation of the marginals and/or the corresponding copulae. These are only some of the possible avenues of future research that will be undertaken.

\bibliography{DRBGAM}
\bibliographystyle{apa}


\clearpage
\begin{center}
\textbf{\LARGE Supplementary Material to \\``Discrete Responses in Bivariate Generalized Additive Models"}
\end{center}
\setcounter{equation}{0}
\setcounter{section}{0}
\setcounter{figure}{0}
\setcounter{table}{0}
\setcounter{algorithm}{0}
\setcounter{page}{1}
\makeatletter
\renewcommand{\thesection}{S.\arabic{section}}
\renewcommand{\theequation}{S\arabic{equation}}
\renewcommand{\thefigure}{S\arabic{figure}}
\renewcommand{\thetable}{S\arabic{table}}
\renewcommand{\thealgorithm}{S\arabic{algorithm}}
\renewcommand{\bibnumfmt}[1]{[S#1]}
\renewcommand{\citenumfont}[1]{S#1}


\section{Asymptotic Behaviour of the Estimator}\label{sct:GAMdiscr.asymptotics}
This section provides some arguments about the asymptotic behaviour of the proposed MPL estimator. Analogous results were also achieved by \cite{Kauermann_05} in the context of survival models, and by \cite{RMW_15} for a bivariate system of dichotomous outcomes. Although our derivations are based on the somehow theoretically stringent assumption that the dimension of the spline bases does not increase with the sample size, this instance is still worth to be considered because, in practice, the bases' dimensions have to be fixed in order to achieve estimation. Nonetheless, by taking the number of the bases relatively rich such to appropriately describe the unknown curves in the model, it is possible to assume heuristically that the approximation bias is negligible compared to the estimation variability (\citealp{Kauermann_05}). To the best of our knowledge, at present the relaxation of this assumption has been confined to the sole analysis of B-splines for their convenient representation and handling as, for instance, did \cite{KKF_09}. Therefore, despite the theoretical relevance of these results, they still do not encompass the whole range of smooths allowed by this work.

Let $\boldsymbol\vartheta_0$ be the ``true'' parameter vector, in the sense that it induces the best approximating likelihood corresponding to the structure that has generated the data. Namely, $\boldsymbol\vartheta_0$ is set the minimiser of the Kullback-Leibler discrepancy
\begin{equation}
\boldsymbol\vartheta_0:=\operatornamewithlimits{\arg\min}_{\boldsymbol\vartheta\in\Theta}\textup{KL}(\mathcal L_t|\mathcal L_n)=\mathbb E[\ell_t-\ell_n(\boldsymbol\vartheta)],\nonumber\label{eq:GAMdiscr.KLdist}
\end{equation}
where the expectation above is carried out with respect to the true model distribution. As a consequence, by direct differentiation of the above equation, we are implicitly defining $\boldsymbol\vartheta_0$ to be the vector such that $\boldsymbol0=\mathbb E[\nabla_{\boldsymbol\vartheta_0}\ell_n]$. For the proceeding analysis we rely on the following regularity conditions:

\begin{assump}
$\nabla_{\boldsymbol\vartheta_0}\ell_n=O_p(n^{1/2})$;\label{asmp:GAMdiscr.grad}
\end{assump}

\begin{assump}
$\mathbb E[\nabla_{\boldsymbol\vartheta_0\boldsymbol\vartheta_0^\top}\ell_n]=O(n)$;\label{asmp:GAMdiscr.FisherInfo}
\end{assump}

\begin{assump}
$\nabla_{\boldsymbol\vartheta_0\boldsymbol\vartheta_0^\top}\ell_n-\mathbb E[\nabla_{\boldsymbol\vartheta_0\boldsymbol\vartheta_0^\top}\ell_n]=O_p(n^{1/2})$; and\label{asmp:GAMdiscr.diffFisher}
\end{assump}

\begin{assump}
$\mathbf S_{\boldsymbol\lambda}=o(n^{1/2})$. Following \cite{Kauermann_05}, this assumption can be equivalently re-stated as $\lambda_{j,m_j}=o(n^{1/2})$ from the very construction of the penalty matrix, and from the fact that its dimensionality is taken fixed as $n$ increases.\label{asmp:GAMdiscr.Smatrix}
\end{assump}

\noindent The above \ref{asmp:GAMdiscr.grad}-\ref{asmp:GAMdiscr.diffFisher} are the standard conditions for the consistency of the unpenalized ML estimator, whereas \ref{asmp:GAMdiscr.Smatrix} ensures that, in the large sample limit, $\mathcal P_{\boldsymbol\lambda}$ becomes irrelevant for the fitting. For a further investigation, we also need an additional condition aimed at describing the behaviour of the log-likelihood third derivatives, and it guarantees the asymptotic Normal distribution of the score:

\begin{assump}
for every $\vartheta^s\in\boldsymbol\vartheta$, $(\partial^3/\partial\vartheta^{s3})\ell_n(\boldsymbol\vartheta)$ exists and satisfies for every point $x\in\mathds R$ and every parameter in the neighbourhood of $\vartheta_0^s$: $|(\partial^3/\partial\vartheta^{s3})\ell_n(\boldsymbol\vartheta)|\le M(x)$, with $\mathbb E[M(x)|\vartheta_0^s]<\infty$; and let $0\le\mathcal I(\vartheta_0^s)<\infty$.\label{asmp:GAMdiscr.3rdDeriv}
\end{assump}

\noindent Then it follows:

\begin{prop}\label{prop:GAMdiscr.MPLEconsistency}
Let $\boldsymbol\vartheta_0$ be the parameter vector defined as in (\ref{eq:GAMdiscr.KLdist}) and assume that conditions \ref{asmp:GAMdiscr.grad}-\ref{asmp:GAMdiscr.3rdDeriv} hold; then the MPL estimator $\widehat{\boldsymbol\vartheta}:=\operatornamewithlimits{\arg\max}_{\boldsymbol\vartheta\in\Theta}[\ell(\boldsymbol\vartheta)-1/2\boldsymbol\vartheta^\top\mathbf S_{\boldsymbol\lambda}\boldsymbol\vartheta]$ satisfies 
\begin{equation}
\widehat{\boldsymbol\vartheta}-\boldsymbol\vartheta_0=\mathbf F^{-1}(\boldsymbol\lambda)(\nabla_{\boldsymbol\vartheta_0}\ell(\boldsymbol\vartheta_0)-\mathbf S_{\boldsymbol\lambda}\boldsymbol\vartheta_0)[1+o_p(1)],\label{eq:GAMdiscr.consistencyMPLE}
\end{equation}
with $\mathbf F^{-1}(\boldsymbol\lambda)=(\mathbf S_{\boldsymbol\lambda}-\mathbb E[\nabla_{\boldsymbol\vartheta_0\boldsymbol\vartheta_0^\top}\ell(\boldsymbol\vartheta_0)])^{-1}$. In particular, the leading stochastic component in (\ref{eq:GAMdiscr.consistencyMPLE}) has asymptotic order $O_p(n^{-1/2})$.
\end{prop}

\begin{proof}
We first set the notation. Let us denote by $\vartheta^j$ the $j$-th component of the parameter vector $\boldsymbol\vartheta=(\vartheta^1,\ldots,\vartheta^p)^\top$, and define subsequently $\ell_{\textup p,j}:=\partial\ell_{\textup p}/\partial\vartheta^j$ the partial derivative of the penalized log-likelihood with respect to $\vartheta^j$; higher order derivatives are denoted subsequently. Also, the ``hat'' notation $\widehat{\ell}_{\textup p}$ stands for $\ell_{\textup p}(\widehat{\boldsymbol\vartheta})$, while the convention of omitting the listing of parameters is used wherever the relevant quantities are evaluated at the best coefficient $\boldsymbol\vartheta_0$, that is $\ell_{\textup p}:=\ell_{\textup p}(\boldsymbol\vartheta_0)$.

Using the Einstein summation convention, we expand $\widehat{\ell}_{\textup p,r}$ around $\ell_{\textup p,r}$ using a second order Taylor approximation:
\begin{equation}
0=\widehat{\ell}_{\textup p,r}=\ell_{\textup p,r}+\ell_{\textup p,rs}(\widehat{\vartheta}-\vartheta_0)^s+\frac{1}{2}\ell_{\textup p,rst}(\widehat{\vartheta}-\vartheta_0)^{st}+\cdots\nonumber
\end{equation}
with $(\widehat{\vartheta}-\vartheta_0)^s:=\widehat{\vartheta}^s-\vartheta_0^s$
and $(\widehat{\vartheta}-\vartheta_0)^{st}=(\widehat{\vartheta}-\vartheta_0)^s(\widehat{\vartheta}-\vartheta_0)^t$. Solving the above equation for $\widehat{\vartheta}-\vartheta_0$, and denoting by superscripts the inverses of the respective quantities, we get (\citealp{BNC_94}):
\begin{equation}
(\widehat{\vartheta}-\vartheta_0)^r=-\ell_{\textup p}^{rs}\ell_{\textup p,s}-\frac{1}{2}\ell_{\textup p}^{rtv}\ell_{\textup p,u}\ell_{\textup p,w}+\cdots\label{eq:GAMdiscr.serInver}
\end{equation}
where $\ell_{\textup p}^{rtv}:=\ell_{\textup p}^{rs}\ell_{\textup p}^{tu}\ell_{\textup p}^{vw}\ell_{\textup p,stv}$, and $\ell_{\textup p}^{rs}$ is the $(r,s)$-th element of the inverse observed (penalized) Fisher Information. Equation (\ref{eq:GAMdiscr.serInver}) can be simplified as follows (see, for example, \citealp{Kauermann_05}): $\ell_{\textup p,rs}:=f_{rs}(\lambda)+r_{rs}$, where $f_{rs}(\lambda):=f_{rs}(0)-s_{\lambda}^{rs}$ is the penalized expected Fisher Information contribution: $f_{rs}(0):=\mathbb E[\partial\ell/\partial\vartheta^r\partial\vartheta^s]$, and $r_{rs}:=\ell_{rs}-f_{rs}(0)$.

Under assumptions \ref{asmp:GAMdiscr.FisherInfo} and \ref{asmp:GAMdiscr.Smatrix} we find that $f_{rs}(\lambda)$ is of asymptotic order $O(n)$, and that $r_{rs}=O_p(n^{1/2})$ directly from \ref{asmp:GAMdiscr.diffFisher}. We can then simplify the first term of (\ref{eq:GAMdiscr.serInver}) as
\begin{eqnarray}
-\ell_{\textup p}^{rs}&=&\mathbb E[\ell_{\textup p,r}\ell_{\textup p,s}]^{-1}+\mathbb E[\ell_{\textup p,r}\ell_{\textup p,t}]^{-1}\mathbb E[\ell_{\textup p,s}\ell_{\textup p,u}]^{-1}(\mathbb E[\ell_{\textup p,t}\ell_{\textup p,u}]+\ell_{p,tu})\nonumber\\
&=&-f^{rs}(\lambda)+f^{rt}(\lambda)f^{su}(\lambda)(-f_{rs}(\lambda)+\ell_{p,tu}),\nonumber
\end{eqnarray}
that is $\ell_{\textup p}^{rs}=f^{rs}(\lambda)-f^{rt}(\lambda)f^{su}r_{tu}$; following now the argument of \cite{KKF_09} we have
\begin{equation}
\ell_{\textup p}^{rs}=f^{rs}(\lambda)[1+O(n^{-1})O_p(n^{1/2})]=f^{rs}(\lambda)[1+O_p(n^{-1/2})].\nonumber
\end{equation}
We need to characterise next the order of $\ell_{\textup p}^{rtv}$, which in turns depend on the one of $\ell_{\textup p,stv}$. First note that $\ell_{\textup p,stv}=\ell_{stv}$ from the very construction of the penalized likelihood estimator, so that we can safely apply \ref{asmp:GAMdiscr.3rdDeriv}, implying that we can bound in probability the third derivative of the log-likelihood. Then, by the strong law of large numbers, we have that, for almost every sequence of $\{x_1,\ldots,x_n\}$ and every $\boldsymbol\vartheta\in\Theta$,
\begin{equation}
|n^{-1}\ell_{stv}|\le n^{-1}\sum_iM(x_i)\xrightarrow{as}\mathbb E[M(x)]\nonumber
\end{equation}
as $n\rightarrow\infty$, hence $n^{-1}\ell_{stv}=O_p(1)$. It is then implied $\ell_{stv}=O_p(n)$ and, after some tedious computations, $\ell_{\textup p}^{rtv}=f^{rs}(\lambda)f^{tu}(\lambda)f^{vw}(\lambda)O_p(n)=O_p(n^{-2})$ so that $\ell_{\textup p}^{rtv}\ell_{\textup p,u}\ell_{\textup p,w}=O_p(n^{-1})$ since $\ell_{\textup p,u}=O_p(n^{1/2})-o(n^{1/2})$. We also find that $\ell_{\textup p}^{rs}\ell_{\textup p,s}$ has order $O_p(n^{-1/2})+o(n^{-1/2})$, that is the second addendum in (\ref{eq:GAMdiscr.serInver}) becomes asymptotically negligible compared to $\ell_{\textup p}^{rs}\ell_{\textup p,s}$. We can then write $(\widehat{\vartheta}-\vartheta_0)^r=-f^{rs}(\lambda)\ell_{\textup p,s}[1+o_p(1)]$, whose leading terms, in matrix notation, are $\mathbf F^{-1}(\boldsymbol\lambda)(\nabla_{\boldsymbol\vartheta_0}\ell(\boldsymbol\vartheta_0)-\mathbf S_{\boldsymbol\lambda}\boldsymbol\vartheta_0)$, from which the assertion follows.

The stochastic order of the above terms then stems from $f^{rs}(\lambda)\ell_{\textup p,s}=O_p(n^{-1/2})+o_p(n^{-1/2})=O_p(n^{-1/2})$.
\end{proof}

\noindent The above derivations also allows us to characterise the bias and the variance of the MPL estimator, as well as their corresponding asymptotic orders. In particular, we find that
\begin{equation}
\mathbb E[\widehat{\boldsymbol\vartheta}]-\boldsymbol\vartheta_0=\mathbf F^{-1}(\boldsymbol\lambda)\mathbf S_{\boldsymbol\lambda}\boldsymbol\vartheta_0[1+o(1)]\label{eq:GAMdiscr.MPLEbias}
\end{equation}
and
\begin{equation}
\mathbb V\textup{ar}[\widehat{\boldsymbol\vartheta}]=-\mathbf F^{-1}(\boldsymbol\lambda)\mathbb E[\nabla_{\boldsymbol\vartheta_0\boldsymbol\vartheta_0^\top}\ell(\boldsymbol\vartheta_0)]\mathbf F^{-1}(\boldsymbol\lambda)[1+o(1)],\label{eq:GAMdiscr.MPLEvariance}
\end{equation}
with orders of $o(n^{-1/2})$ and $O(n^{-1})$, respectively. In fact, we immediately obtain the (asymptotic) equivalence
\begin{equation}
(\widehat{\vartheta}-\vartheta_0)^r=-f^{rs}(\lambda)\ell_{\textup p,s}[1+o_p(1)]\nonumber
\end{equation}
from which
\begin{equation}
\mathbb E[(\widehat{\vartheta}-\vartheta_0)^r]=-\mathbb E[f^{rs}(\lambda)(\ell_s-s_{\lambda}^s\vartheta_0^s)][1+o(1)]=f^{rs}s_{\lambda}^s\vartheta_0^s[1+o(1)]\nonumber
\end{equation}
and
\begin{equation}
\mathbb V\textup{ar}[\widehat{\vartheta}^r]=(f^{rs}(\lambda))^2\mathbb V\textup{ar}[\ell_s][1+o(1)]=-(f^{rs}(\lambda))^2f_{rs}(0)[1+o(1)].\label{eq:MPLEvarProof}
\end{equation}
Finally, invoking \ref{asmp:GAMdiscr.FisherInfo} and \ref{asmp:GAMdiscr.Smatrix}, and since $f^{rs}(\lambda)$ is $O(n^{-1})$, we compute $\mathbb E[(\widehat{\vartheta}-\vartheta_0)^r]=O(n^{-1})o(n^{1/2})=o(n^{-1/2})$, while $\mathbb V\textup{ar}[\widehat{\vartheta}^r]$ is led by terms of order $O(n^{-2})O(n)=O(n^{-1})$.

\section{Confidence Intervals Computation}\label{sctSM:GAMdiscr.CI}
At convergence of the estimation algorithm, the penalised GLS representation of the model induces a covariance matrix of the estimator of the form $\mathbf V_{\widehat{\boldsymbol\vartheta}}=\mathcal H_{\textup p}^{-1}\mathcal H\mathcal H_{\textup p}^{-1}$ which can in principle be used to compute the standard errors of each component of $\widehat{\boldsymbol\vartheta}$. An appealing alternative approach to conduct inference, however, is to advocate a Bayesian reasoning as based on the posterior distribution of $\boldsymbol\vartheta|\mathbf w$
\begin{equation}
\boldsymbol\vartheta|\mathbf w\sim\mathcal N_p([\mathbf{D}^\top\overline{\mathbf W}\mathbf D+\mathbf S_{\boldsymbol\lambda}]^{-1}\mathbf{D}^\top\overline{\mathbf W}\mathbf z,[\mathbf{D}^\top\overline{\mathbf W}\mathbf D+\mathbf S_{\boldsymbol\lambda}]^{-1}),\nonumber
\end{equation}
which is equivalent to choose $\boldsymbol\vartheta|\mathbf v\sim\mathcal N_p(\widehat{\boldsymbol\vartheta},\mathbf V_{\boldsymbol\vartheta})$, with $\mathbf V_{\boldsymbol\vartheta}:=-\mathcal H_{\textup p}^{-1}$. The Bayesian framework above emerges naturally from the specification of the model through a roughness penalty approach. In effect, as \cite{Wahba_83} and \cite{Silverman_85} recognised, the imposition of any kind of penalisation in the estimating procedure corresponds to the explication of some kind of prior beliefs about the likely features of the true model. Specifically, the definition of a normal prior for $\boldsymbol\vartheta$, $f_{\boldsymbol\vartheta}\propto\exp\{-1/2\boldsymbol\vartheta^\top\mathbf S_{\boldsymbol\lambda}\boldsymbol\vartheta\}$, implies that smoother models are more likely than wiggly ones, while it gives equal probability density to all models of equal smoothness (\citealp{Wood_06ANZJS}). The stated posterior distribution is then a consequence of the asymptotic normality of $\mathbf w:=\mathbf D^\top\mathbf W\mathbf z$. Upon re-writing $\mathbf w=\mathbf D^\top\mathbf W\mathbf D\boldsymbol\vartheta+\mathbf D^\top\mathbf u$, it holds that the last addendum is aymptotically bounded by a random vector with distribution $\mathcal N_p(\boldsymbol0_p,\mathbf D^\top\overline{\mathbf W}\mathbf D)$ because of \ref{asmp:GAMdiscr.3rdDeriv}, whereas the first one converges in probability to $\mathbf D^\top\overline{\mathbf W}\mathbf D\boldsymbol\vartheta$ from which the desired distribution follows.

For the construction of confidence intervals of the non-parametric model components, the employment of $\mathbf V_{\boldsymbol\vartheta}$ is usually preferred to $\mathbf V_{\widehat{\boldsymbol\vartheta}}$. In fact, as argued by \cite{MarraWood_12} in the context of GAMs, the former can produce intervals with close to nominal ``across-the-function'' frequentist coverage probabilities, as resulting from the inclusion in $\mathbf V_{\boldsymbol\vartheta}$ of both a bias and a variance component, a feature which is not shared instead by $\mathbf V_{\widehat{\boldsymbol\vartheta}}$. A key requirement at the basis of the result is that the magnitude of the bias component is substantially of a small portion compared to the sampling variability, an occurrence that is guaranteed wherever heavily over-smoothing is prevented (\citealp{Nychka_88}). Point-wise confidence intervals for the estimated non-parametric curve $\widehat{s}_{j,l_j}$ are then obtained from $\mathcal N(s_{j,l_j}(\textup v_{j,l_j,i}),\boldsymbol b_{j,l_j,i}^\top\mathbf V_{\boldsymbol\vartheta,[j,l_j]}\boldsymbol b_{j,l_j,i})$, where $\mathbf V_{\boldsymbol\vartheta,[j,l_j]}$ is the sub-matrix of $\mathbf V_{\boldsymbol\vartheta}$ corresponding to the parameters associated to the $(j,l_j)$-th smooth. 

More generally, confidence intervals for a non-linear function $T(\widehat{\boldsymbol\vartheta})$ of the MPLE can be constructed by a convenient simulation scheme from the posterior distribution $\boldsymbol\vartheta|\mathbf w$, as illustrated in the pseudo-code given in Algorithm \ref{alg:GAMdiscr.CI}. 

To conclude, the asymptotic equivalence between the frequentist and the Bayesian variance estimators can be establised as follows:

\begin{algorithm}[t]
\caption{Approximate $(1-\alpha)\%$ Confidence Interval (CI$_{1-\alpha}$) for $T(\widehat{\boldsymbol\vartheta})$} 
\label{alg:GAMdiscr.CI}      
\begin{algorithmic} 
\REQUIRE $\widehat{\boldsymbol\vartheta}$; $\mathbf D$; $\mathbf W$; $\mathbf S_{\boldsymbol\lambda}$; $N_{sim}$
\STATE{$\mathbf V_{\boldsymbol\vartheta}\Leftarrow(\mathbf D^\top\mathbf W\mathbf D+\mathbf S_{\boldsymbol\lambda})^{-1}$}\vspace{0.1cm}
\STATE{draw $\{\boldsymbol\vartheta_r^*\}_{r=1}^{N_{sim}}\sim\mathcal N_p(\widehat{\boldsymbol\vartheta},\mathbf V_{\boldsymbol\vartheta}(\widehat{\boldsymbol\vartheta}))$}\vspace{0.1cm}
\STATE{$\{T_r^*\}\Leftarrow\{T(\boldsymbol\vartheta_r^*)\}_r$}\vspace{0.1cm}
\STATE{$T^*:=\{T_{(1)}^*,\ldots,T_{(N_{sim})}^*\}$ such that $T^*_{(r)}\le T^*_{(r+1)}$, $\forall r=,1\ldots,N_{sim}$}\vspace{0.1cm}
\STATE{define $T^*_{\alpha}$ as the smallest $[N_{sim}\alpha]$-th value of $T^*$}\vspace{0.1cm}
\STATE{$\textup{CI}_{1-\alpha}\Leftarrow[T_{\alpha/2}^*,T_{1-\alpha/2}^*]$}
\end{algorithmic}
\end{algorithm}

\begin{cor}[to Proposition \ref{prop:GAMdiscr.MPLEconsistency}]\label{cor:GAMdiscr.varEquiv}
Under assumptions \ref{asmp:GAMdiscr.grad}-\ref{asmp:GAMdiscr.3rdDeriv} the re-scaled frequentist and Bayesian variance estimators $\sqrt{n}\mathbf V_{\widehat{\boldsymbol\vartheta}}$ and $\sqrt{n}\mathbf V_{\boldsymbol\vartheta}$, respectively, converge to the same quantity
\begin{equation}
-\sqrt{n}\mathbb E^{-1}[\nabla_{\boldsymbol\vartheta_0\boldsymbol\vartheta_0^\top}\ell(\boldsymbol\vartheta_0)]\nonumber
\end{equation}
as $n\rightarrow\infty$. 
\end{cor}

\begin{proof}
Using equation (\ref{eq:MPLEvarProof}) above, we derive
\begin{equation}
\mathbb V\textup{ar}[\widehat{\vartheta}^r]=(f^{rs}(\lambda))^2\mathbb V\textup{ar}[\ell_s][1+o(1)]=-(f^{rs}(\lambda))^2f_{rs}(0)[1+o(1)],\nonumber
\end{equation}
from which 
\begin{eqnarray}
\sqrt{n}\mathbb V\textup{ar}[\widehat{\vartheta}^r]&=&-\sqrt{n}(f_{rs}(0)-s_{\lambda}^{rs})^{-2}f_{rs}(0)[1+o(1)]\nonumber\\
&\approx &-\sqrt{n}^{-1}(\sqrt{n}^{-1}f_{rs}(0)-o(1))^{-2}f_{rs}(0)\approx-\sqrt{n}(f^{rs}(0))^2f_{rs}(0)=-\sqrt{n}(f_{rs}(0))^{-1};\nonumber
\end{eqnarray}
this corresponds immediately to the statement once written in matrix notation. Similarly, we can compute the asymptotic limit of the Bayesian variance estimator: reminding that $\mathbf V_{\boldsymbol\vartheta}=-(\nabla_{\boldsymbol\vartheta\boldsymbol\vartheta^\top}\ell_{\textup p}(\boldsymbol\vartheta))^{-1}$, analogous arguments as above lead to:
\begin{equation}
-\sqrt{n}\ell_{\textup p}^{rs}=-\sqrt{n}f^{rs}(\lambda)[1+o(1)]\approx-\sqrt{n}(f_{rs}(0))^{-1}\nonumber
\end{equation}
which concludes the proof.
\end{proof}

\section{Proof of Proposition \ref{prop:GAMdiscr.AIC}}\label{sctSM:GAMdiscr.AICproof}
Let us consider first a Taylor expansion of $-2\ell(\boldsymbol\vartheta^*)$ about $\boldsymbol\vartheta$:
\begin{equation}
-2\ell(\boldsymbol\vartheta^*)\approx-2\ell(\boldsymbol\vartheta)-2(\boldsymbol\vartheta^*-\boldsymbol\vartheta)^\top \nabla_{\boldsymbol\vartheta}\ell(\boldsymbol\vartheta)-(\boldsymbol\vartheta^*-\boldsymbol\vartheta)^\top\nabla_{\boldsymbol\vartheta\boldsymbol\vartheta^\top}\ell(\boldsymbol\vartheta)(\boldsymbol\vartheta^*-\boldsymbol\vartheta),\nonumber
\end{equation}
and recall that, in the large sample approximation, $-\nabla_{\boldsymbol\vartheta\boldsymbol\vartheta^\top}\ell(\boldsymbol\vartheta)=\mathbf D^\top\mathbf W\mathbf D\stackrel{p}{\longrightarrow}\mathbf D^\top\overline{\mathbf W}\mathbf D$, hence we can write the addenda in the above expression as:
\begin{equation}
(\boldsymbol\vartheta^*-\boldsymbol\vartheta)^\top\mathbf D^\top\overline{\mathbf W}\mathbf D(\boldsymbol\vartheta^*-\boldsymbol\vartheta)=\big\|\sqrt{\overline{\mathbf W}}(\overline{\mathbf z}-\mathbf D\boldsymbol\vartheta^*)\big\|^2+\big\|\sqrt{\overline{\mathbf W}}^{-1}\mathbf u\big\|^2-2\left\langle\overline{\mathbf z}-\mathbf D\boldsymbol\vartheta^*;\mathbf u\right\rangle\nonumber
\end{equation}
and
\begin{equation}
(\boldsymbol\vartheta^*-\boldsymbol\vartheta)^\top \mathbf D^\top\mathbf u=\big\|\sqrt{\overline{\mathbf W}}^{-1}\mathbf u\big\|^2-\left\langle\overline{\mathbf z}-\mathbf D\boldsymbol\vartheta^*;\mathbf u\right\rangle.\nonumber
\end{equation}
Then we have
\begin{equation}
-2\ell(\boldsymbol\vartheta^*)\approx-2\ell(\boldsymbol\vartheta)-\big\|\sqrt{\overline{\mathbf W}}^{-1}\mathbf u\big\|^2+\big\|\sqrt{\overline{\mathbf W}}(\overline{\mathbf z}-\mathbf D\boldsymbol\vartheta^*)\big\|^2\nonumber
\end{equation}
and, by noticing that the smoothing parameter vector affects the latter approximation only through the updated iteration $\boldsymbol\vartheta^*$, and that we are interested in optimising a criterion with respect to $\boldsymbol\lambda$, it is lecit to drop any addendum not depending on it. Hence, one can indifferently consider an equivalent UBRE criterion given by
\begin{equation}
\mathcal V_u\propto\big\|\sqrt{\overline{\mathbf W}}(\overline{\mathbf z}-\mathbf D\boldsymbol\vartheta^*)\big\|^2+2\textup{tr}(\mathbf P)\equiv2\textup{tr}(\mathbf P)-2\ell(\boldsymbol\vartheta^*).\nonumber
\end{equation}\hspace*{\fill}$\qedsymbol$

\section{Data Generating Process Employed in Figure \ref{fig:GAMdiscr.simSmooths}}\label{appx:GAMdiscr.DGP}
The simulation results depicted in Figure \ref{fig:GAMdiscr.simSmooths} comprises a bivariate system of equations specified by the following Data Generating Process (DGP):
\begin{equation}
\begin{array}{ccr}
y_{1,i}^*&=&\textup x_{1,i}+2\textup x_{2,i}+\textup x_{3,i}+s_{1,1}(\textup v_{1,i})+s_{1,2}(\textup v_{2,i})+\epsilon_{1,i}\\
y_{2,i}^*&=&-0.3 y_{1,i}^* +\textup x_{1,i}-2\textup x_{2,i}+s_{2,1}(\textup v_{1,i})+\epsilon_{2,i}
\end{array}\mbox{ }\mbox{ }\mbox{ }\boldsymbol\epsilon_i\sim\mathcal N_2\left(\left(\begin{array}{c}0\\0\end{array}\right),\left[\begin{array}{cc}1 & 0.9\\0.9 & 1\end{array}\right]\right)\nonumber
\end{equation}
for a sample size of $10,000$ observations, and $N=100$ replications. The test functions are displayed in red in the figure, and given by $s_{1,1}(\textup v_{1,i})=1-\textup v_{1,i}+1.6\textup v_{1,i}^2-\sin(5\textup v_{1,i})$, $s_{1,2}(\textup v_{2,i})=4\textup v_{2,i}$ and $s_{2,1}(\textup v_{1,i})=0.08\{\textup v_{1,i}^{11}[10(1-\textup v_{1,i})]^6+10(10\textup v_{1,i})^3(1-\textup v_{1,i})^{10}\}$. Furthermore, the ordered values of $y_{j,i}$ have been computed following the observation rule:
\begin{equation}
y_{j,i}=\sum_{k_j\in\mathcal K_j}k_j\mathds 1_{c_{j,k_j-1}<y_{j,i}^*\le c_{j,k_j}},\nonumber
\end{equation}
for every $j\in\{1,2\}$, and obtained by setting the threshold parameters at $\boldsymbol c_1:=(-2,-1, 0,2)^\top$ and $\boldsymbol c_2:=(-1.4,-0.7,-0.2,0.7,3)^\top$.

\end{document}